\documentclass[preprint,12pt]{elsarticle}

\usepackage{amssymb}
\usepackage{amsthm}
\usepackage{amsmath}
\usepackage{dsfont}

\DeclareMathOperator*{\argmin}{arg\,min}
\usepackage{caption}
\usepackage{subcaption}
\usepackage{float}

\usepackage{amsfonts}

\usepackage{xcolor}

\newtheorem{prop}{Proposition}

 \usepackage{algpseudocode}
 \usepackage{lscape}
 
 \usepackage{bm}
 \newcommand{\x}{{\bm x}}
 \newcommand{\nub}{{\bm \nu}}
 \newcommand{\N}{{\bm N}}

 \newcommand{\red}[1]{\textcolor{black}{#1}}
 
  \newcommand{\lch}[1]{\textcolor{black}{#1}}
 
\begin{document}

\begin{frontmatter}

\title{Sequential detection of a temporary change in multivariate time series}

\author[inst1,inst2]{Victor Watson}

\affiliation[inst1]{organization={CEA, DAM, DIF},
            addressline={~}, 
            city={Arpajon},
            postcode={ F-91297}, 
            country={France}}

\affiliation[inst2]{organization={Univ Bretagne Sud, CNRS UMR 6205, LMBA},
            addressline={Rue Andre Lwoff}, 
            city={Vannes},
            postcode={F-56000}, 
            country={France}}
            
\author[inst2]{François Septier}

\author[inst1]{Patrick Armand}

\author[inst1]{Christophe~Duchenne}

\begin{abstract}

In this work, we aim to provide a new and efficient recursive detection method for temporarily monitored signals.  Motivated by the case of the propagation of an event \lch{over} a field of sensors, we \red{assumed} that the change in the statistical properties in the monitored signals can only be temporary. Unfortunately, to our best knowledge, existing recursive and simple detection techniques such as the ones based on the cumulative sum (CUSUM) do not consider the temporary aspect of the change in a multivariate time series. In this paper, we propose a novel simple and efficient sequential detection algorithm, named Temporary-Event-CUSUM (TE-CUSUM). \red{By combining} with a new adaptive way to aggregate local CUSUM variables from each data stream, we empirically show that the TE-CUSUM has a very good detection rate in the case of an event passing through a field of sensors in a very noisy environment.

\end{abstract}


\end{frontmatter}

\section{Introduction}

The multiplication of industrial sites near populated areas increases the danger for populations in cases of an unexpected release of an hazardous compound. Densely populated areas can also be a target for ill-intention people \lch{who would} release some toxic material and cause many victims. In such cases early detection can be crucial. If these sensitive areas are monitored, waiting for the level of toxic compound to be sufficient so that it is unambiguously monitored by the sensors can have the consequence of being too late in one's response to a threat. Sequential change-point detection uses the statistics of a data stream to detect an abnormality in a signal while the signal is still low. This means that one could detect a small concentration of a chemical in the air before a level threatening human life is reached\lch{. One would detect} an abnormal radioactivity level due to an irradiating source, for instance contained in a dirty bomb \lch{before explosion}. These sequential detection techniques can be used to detect the presence of a pollutant in the air such as in \cite{rajaona2015}. Early detection allows estimation techniques such as in \cite{septier20} to start monitoring the data at the right moment and facilitates the convergence to a solution while decreasing the computational cost. It is also used for early seismic detection \cite{Popescu2014}, or early detection of infected people during a pandemic \cite{Braca2021}, but can also be applied to many other fields such as \cite{Aue2011} and \cite{Shbat2019}. 

The CUSUM (CUmulative SUM) technique \cite{Page54} is a powerful univariate sequential change-point detection tool on which are based many detection techniques. The extension of the CUSUM to multivariate cases is not strait-forward and has been the object of many considerations in the process-control community (see \cite{Golosnoy2009,Mei2010,Kurt2018,Tartakovsky14_GLRT,rovatsos2020a,Banerjee2015,Xie2013}). Moreover, the issue of the possible non-synchronicity of the monitoring between the sensors in multivariate cases for the detection of temporary change remains a research open question. Indeed, the existing ways to deal with these problems require to lose the recursive computation of the test statistics necessary to trigger detection. The temporary change is not a common consideration in the process-control community as when a process gets out-of-control, it rarely gets back in-control. When we extend the sequential detection technique to some other physical problems such as the ones cited earlier, this back in-control scenario is what we expect as a sensor can be exposed only for a limited duration.

In this paper we propose a new multivariate CUSUM-based technique to deal with temporary changes without losing the recursive computation. Indeed, the Temporary-Event-CUSUM (TE-CUSUM) does not require the change to be permanent or synchronous between data streams to be detected. \red{We believe this novel TE-CUSUM will be relevant to many application cases such as when a pollutant released into the air is dispersed through a wide area and low concentrations are recorded by different sensors at different times. To our best knowledge existing methods would require simultaneous sensors' exposures or would be too computationally expensive methods for their online practical use. }
We also developed a new adaptive method to combine the local test statistics so it increases the performances of the TE-CUSUM when the subset of sensor which is affected with the signal is unknown.

This paper is organised as follows\red{. In} the second section, we introduce our detection problem and the CUSUM technique as a Generalised-likelihood ratio test (GLRT). We set the principle of the proposed TE-CUSUM and show its equivalence with the CUSUM for univariate cases. In the third section, we extend the model problem to multivariate temporary events and make a quick review of the existing methods for multivariate sequential change-point detection. We then introduce a new strategy \red{that combines} local statistics with a novel adaptive censoring method and also demonstrate the efficiency of the TE-CUSUM in multivariate sequential detection cases. The fourth section is a validation test in which we compare the efficiency of the different techniques for monitoring the dispersion of a pollutant over a field of sensors.

\section{Sequential change-point detection in univariate time series}

The change-point detection problem in univariate time series can be formulated as the following hypothesis test: 
\begin{align}
    \begin{split}
        {\cal{H}}_0:&p(x_1,\cdot\cdot\cdot,x_n) =  \prod_{k=1}^n f_0(x_k|\theta_0) \\
          {\cal{H}}_1:&p(x_1,\cdot\cdot\cdot,x_n;\nu) = \prod _{k=1}^\nu f_0(x_k|\theta_0)  \prod _{k=\nu+1}^n f_1(x_k|\theta_1)
    \end{split}\label{eq3}
\end{align}

This represents a two case scenario, the first one marked by the hypothesis ${\cal{H}}_0$ for which every sample $x_k$ with $k\leq n$ follows $f_0(x_k|\theta_0)$, the second one stating that there is a time \red{$\nu \in [1 ; n]$} such that $x_k$ with $(k>\nu)$ starts to follow $f_1(x_k|\theta_1)$.

This leads us to the associate likelihood ratio:
\begin{equation}
    \Lambda_n^\nu = \prod_{k =\nu+1}^n \frac{f_1(x_k|\theta_1)}{f_0(x_k|\theta_0)} 
    \label{eq4}
\end{equation}
 Comparing $\Lambda_n^\nu$ to a threshold allows us to define a statistical test sequentially computed to decide between the two hypothesis.

\subsection{Generalised likelihood ratio test and CUSUM}

One problem about the likelihood ratio test of Equation (\ref{eq4}) is that the knowledge of the change-point $\nu$ is needed. In such a situation (unknown parameter in the likelihood distribution), it is common to use a generalised likelihood ratio test (GLRT) \cite{Trees} which is defined in our problem as: 
\begin{equation}
   V_n = \max_{0 \leq \nu < n}(\Lambda_n^\nu) = \max_{0 \leq \nu < n} \prod_{k=\nu+1}^n L_k 
   \label{eq5}
\end{equation}
with,
\begin{equation}
    L_k = \frac{f_1(x_k|\theta_1)}{f_0(x_k|\theta_0)}
\end{equation}
while the change point can be estimated with:
\begin{equation}
   \hat{\nu} = \underset{0 \leq \nu < n}{\operatorname{argmax}} \prod_{k=\nu+1}^n L_k 
   \label{eq6}
\end{equation}

The criterion $V_n$ increases when $L_k > 1$ and decreases when $L_k < 1$. If $k > \nu$, the hypothesis ${\cal{H}}_{1}$ is true and the ratio $L_k$ has a better chance to be greater than 1. $V_n$ will overall increase even if monotony is far from guaranteed. We can then compare $V_n$ to a threshold to trigger detection sequentially once a novel observation is received. In such a sequential setting, since we are interested in the quickest detection method, it is also important to consider the detection delay \red{time} instead of just the probability of detection \cite{Tartakovsky14_GLRT}. 

The generalised likelihood ratio in Equation (\ref{eq5}) can be rewritten in the following recursive form which allows its integration in online systems:
\begin{equation}
    V_n = \max(1,V_{n-1})L_n, ~~~ n\geq 1, V_0=1 
    \label{eq7}
\end{equation}

\subsubsection{CUSUM principle}

The CUSUM technique was first introduced by Page in 1954 \cite{Page54}. This algorithm has been proposed in order to optimise both the detection delay and the average run-length to false alarm (ARL2FA) which is the average time between two false alarms \cite{Tartakovsky14_GLRT}. \red{To compute the ARL2FA, we used the cyclical steady-state ARL decribed by \cite{knoth21}}.  This procedure can be seen as a sequential algorithm to recursively compute the GLRT defined by Equation (\ref{eq5}). By using the log transform of Equation (\ref{eq7}), the CUSUM test statistic is indeed simply given by:
\begin{align}
    \begin{split}
        W_n &= \max(0,W_{n-1})+\log(L_n) , ~~~ n\geq 1, W_0=0 
    \end{split}\label{eq8}
\end{align}

By just computing a sum at each time sample and comparing $W_n$ to a threshold, one can have a robust online detection technique. The question of estimating the change point $\nu$ can be solved easily by expanding Equation (\ref{eq6}) as:
\begin{align}
    \begin{split}
        \hat{\nu} &= \underset{0 \leq \nu < n}{\operatorname{argmax}} \prod_{k=\nu+1}^n L_k\\
        &= \underset{0 \leq \nu < n}{\operatorname{argmax}} \sum_{k=\nu+1}^n \log(L_k), \\
         & =\underset{0 \leq \nu < n}{\operatorname{argmax}} \sum_{k=1}^n \log(L_k)- \sum_{k'=1}^\nu \log(L_{k'}) \\
         &=\underset{0\leq k<n}{\operatorname{argmin}}(S_k). 
    \end{split}
    \label{eq9}
\end{align}
with:
\begin{equation}
    S_k = \sum_{i=1}^k \log \frac{f_1(x_i|\theta_1)}{f_0(x_i|\theta_0)},~ \text{and} ~ S_0 = 0.
    \label{eq9_1}
\end{equation}

\subsubsection{Example}

To illustrate the CUSUM, let us consider a change of mean in a single data stream composed of independent Gaussian random variables: 
\begin{equation}
    \begin{array}{cc}
        f(x_k |\theta_0) & = {\cal{N}}(x_k; \mu_0, \sigma^2) \\
        f(x_k |\theta_1) & = {\cal{N}}(x_k; \mu_1, \sigma^2)
    \end{array}
    \label{eq10}
\end{equation}
In this case the CUSUM test statistic can be computed as:
\begin{equation}
    W_n = \max\left(0, W_{n-1} + \frac{\delta_\mu}{\sigma^2}\left(x_n -\mu_0 -\frac{\delta_\mu}{2}\right)\right)
\end{equation}
\red{where},
\begin{equation}
    \delta_\mu = \mu_1-\mu_0
\end{equation}
\begin{figure}[H]
\hspace{0cm}
    \includegraphics[width=1\hsize]{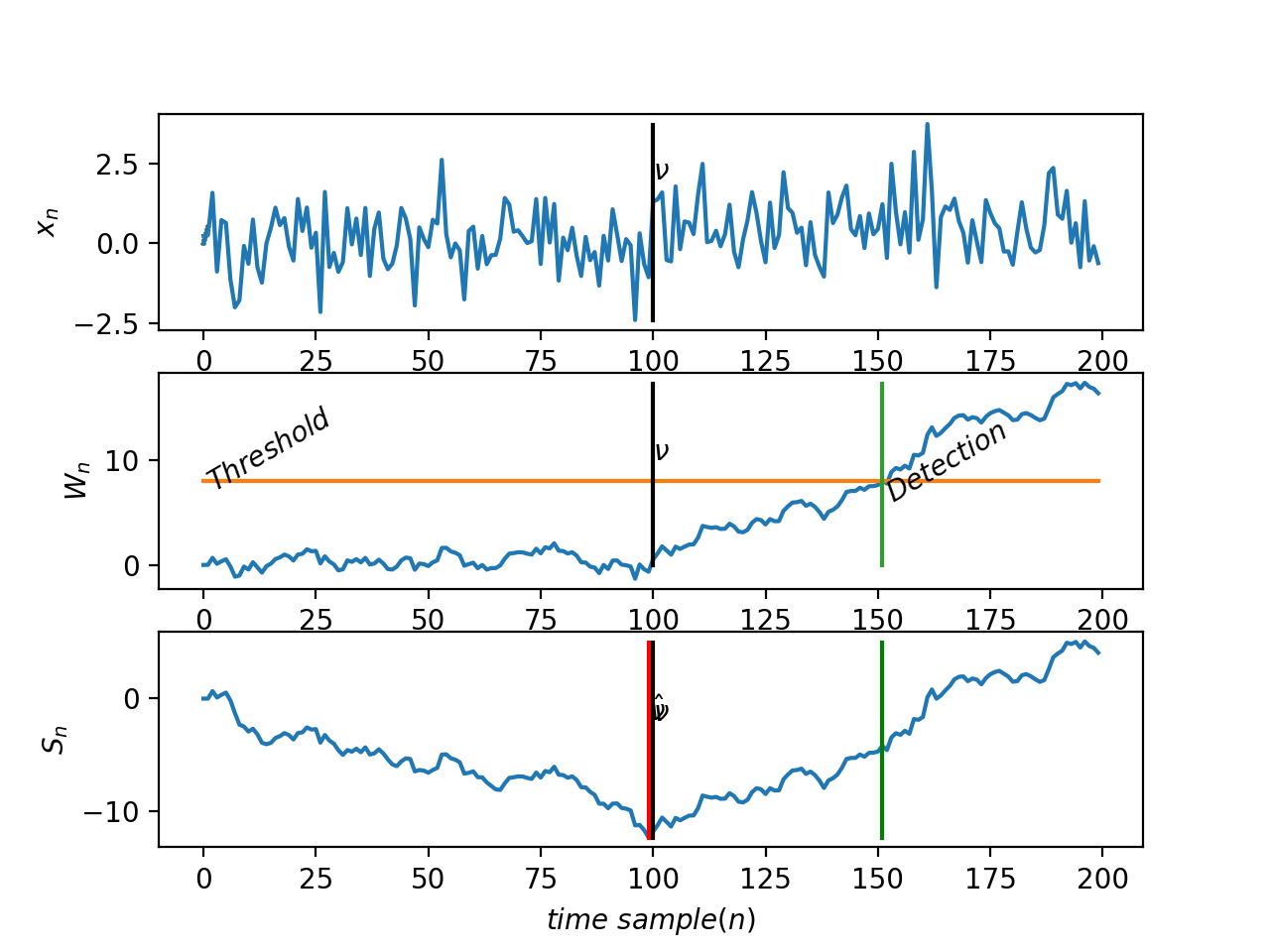}
    \caption{Example of detection of a change of mean in a single data stream  from $\mu_0 = 0$ to $\mu_1 = 0.1$ with a Gaussian noise distribution of standard deviation $\sigma = 1$. The data stream is plotted at the top with the change-point time indicated by a vertical line. The CUSUM variable Wn is plotted on the middle figure with the value of the chosen threshold indicated by the horizontal red line and the time of detection indicated by the green vertical line. The variable from Equation (\ref{eq9_1}) is plotted at the bottom with the estimated change point indicated by the vertical red line.}
    \label{fig4}
\end{figure} 

Figure \ref{fig4} empirically shows that the CUSUM technique is able to detect a change in the mean of a signal which is not obvious by looking only \red{at} the time series $x_n$. \red{The method has a delay} to detect the change-point (51 \lch{time samples} in this case) but it is able to detect it nonetheless. At the bottom of the figure we can see that $\argmin(S_n)$, defined by Equation (\ref{eq9_1}), gives us an estimate of the change point. \red{This estimate can also be found at the last time $W_n = 0$} .  
The CUSUM can be used to detect any changing parameter \cite{Lee2003} even if it is most commonly used to detect a change of mean or variance. As in any detection technique there is a balance to make between detection rate and false alarm, here between the average detection delay and the ARL2FA. A way to deal with the setting of the method is to determine what ARL2FA is tolerable, set the detection threshold (to which $W$ is compared) to get the wanted ARL2FA and then check what average detection delay is obtained. When comparing several methods, one can set the thresholds so that the ARL2FA is the same among all of them and compare the average detection delays to determine which gives the quickest detection.

\subsection{Finite moving average (FMA)}
Concerned by the cases in which the change is temporary, \cite{Tartakovsky_2021} proposes a method to detect this change by computing a likelihood ratio test on a moving time window of the signal defined as:
\begin{equation}
    Z_n^w = \sum_{k = n-w+1}^n S_k
    \label{eq10_1}
\end{equation}
 which depends on the window length $w$. This test statistic, $Z_n^w$ is then compared to a threshold to trigger or not a detection. In \cite{Tartakovsky_2021}, the authors compared this approach to the CUSUM when the amplitude of a change in mean is lower than expected and when the change duration is finite or even if the change is intermittent. These characteristics are of the utmost interest for our purpose.

This technique \red{is similar to the MOSUM (moving sum) \cite{noonan20} and} requires to either memorise $w$ values of the likelihood ratio $S_k$ or to compute $w$ times more operations every time sample than the CUSUM technique. Also, this method seems to be sensitive to the difference between the length of the window and the duration of the signal to detect.   

\subsection{Temporary-Event-CUSUM}

In this paper we introduce a new technique called Temporary-Event-CUSUM (TE-CUSUM). 
Because the change is transitory, the model we consider is defined through the two hypotheses:
\begin{align}
    \begin{split}
        {\cal{H}}_0:&p(x_1,\cdot\cdot\cdot,x_n) =  \prod_{k=1}^n f_0(x_k|\theta_0) \\
         {\cal{H}}_1:&p(x_1,\cdot\cdot\cdot,x_n;\nu,N) = \prod _{k=1}^\nu f_0(x_k|\theta_0)  \prod _{k=\nu+1}^N f_1(x_k|\theta_1) \prod _{k=N+1}^n f_0(x_k|\theta_0)
    \end{split}\label{eq3_2}
\end{align}

\begin{prop}\label{prop_TransSigCUSUM}
The test statistic obtained by solving the generalised likelihood ratio test of Equation (\ref{eq3_2}) can be recursively obtained as follows:
\begin{equation}
    G_n = \max(G_{n-1},W_n), G_0 =0, 
    \label{eq15}
\end{equation}
with $W_n$ being the CUSUM test statistic defined in Equation (\ref{eq8}).
\end{prop}
\begin{proof}
The GLRT of Equation (\ref{eq3_2}) can be written as:
\begin{align}
    \begin{split}
        G_n  & =\max_{[\nu,N]}\sum_{k=\nu}^{N \leq n} \log \frac{f_{1}(X_{k}|\theta_{1})}{f_{0}(X_{0}|\theta_{0})}\\
        & =\max_{0<k \leq n} \left( \max_\nu \sum_{i=\nu}^k \log \frac{f_1(X_i|\theta_0)}{f_0(X_i|\theta_1)} \right)\\
        &= \max_{0\leq k <n} (W_k)
    \end{split}\label{eq14}
\end{align}
which leads straightforwardly to the recursive form introduced in Prop. \ref{prop_TransSigCUSUM}.
\end{proof}

In Equation (\ref{eq14}) it is implicit that $\nu$ is the last change-point before $k$. Moreover, causality forces $\hat{\nu} < n$ when $W_n$ is computed.

On univariate cases, TE-CUSUM is strictly equivalent to the standard CUSUM because a test on $G_n$ is equivalent to a test on $W_n$.
From Equation (\ref{eq15}), we have $G_n = \max_{0 \leq k \leq n}(W_k)$. Therefore if $W_k > G_{k-1}$ then $G_k = W_k$ if $W_k > h$ then $G_k > h$, $h$ being the detection threshold. If $G_{k-1} < h$ and $W_k < G_{k-1}$ then $G_k = G_{k-1} < h$ and $W_k < G_k < h$. There is no way $W_k > h$ without $G_k > h$ and neither $G_k > h | G_{k-1} < h$ without $W_k > h$.

\section{Multivariate detection}

In this section, we consider a multi-sensor network which consists of a collection of indexed sensors ${\cal K}=\left\{1,\ldots,L\right\}$  where each of them observes a realization from the previously discussed model. More specifically, under normal conditions, the distribution which governs the behaviour of each of the sensor is given by $f_0(\cdot|\theta_0)$. At random time and during some random duration, a change could occur which affects a subset of sensors ${\cal K}_c \in \left\{\emptyset\cup {\cal K}\right\}$. The detection problem of a change can be thus formulated using the following binary hypothesis test model:
\begin{align}
    \begin{split}
        {\cal{H}}_0:&p(\x_1,\cdots,\x_n) =  \prod_{l\in{\cal K}}\prod_{k=1}^n f_0(x_{k,l}|\theta_0) \\
          {\cal{H}}_1:&p(\x_1,\cdots,\x_n;\nub,\N) =  \prod_{l\in{\cal K}_c } \prod_{{1\leq k\leq \nu_l}\atop {N_l<k\leq n}} f_0(x_{k,l}|\theta_0)  \prod _{k=\nu_l+1}^{N_l} f_1(x_{k,l}|\theta_1)\\
         & \hspace*{4.5cm} \times\prod_{j\in{\cal K} \setminus {\cal K}_c  }\prod_{k=1}^n f_0(x_{k,j}|\theta_0) 
    \end{split}\label{GeneralMultivariateModel}
\end{align}

\subsection{A brief review of existing procedures}
Solving the problem set by Equation (\ref{GeneralMultivariateModel}) would necessitate to test all combination of change-points for each possible subset of sensors. Some optimisation approaches have been used by \cite{Kurt2018}, or \cite{rovatsos2020a} and \cite{rovatsos2020b}. The major drawback of these approaches is that it becomes rapidly too  computationally expensive and it looses the possibility of a recursive computation. \cite{Golosnoy2009} sees the multivariate CUSUM variable as the norm of the sum of the local test statistics and \cite{Xie2013} takes into account the case where the size of the subset of sensors is roughly known. In some cases, when the number of sensors is very large \cite{Banerjee2015} proposed to only merge binary units to the decision center so that detection is triggered by their number and not by an aggregation of local values. 

Except for \cite{rovatsos2020a} and \cite{rovatsos2020b}, all the developed methods consider that the change is permanent ($N_l = \infty$) for all the affected sensors. \red{The Multi-dimension exponentially weighted moving average (MD-EWMA) algorithm presented by \cite{wu21} takes into account temporary changes but requires the change to be synchronous between the sensors.}

In this work, we intend to remove the limitation $N_l = \infty$ (thus allowing some sensors to stop monitoring the change at some point) without increasing the computational power required for the detection.

Two basic ways to adapt the CUSUM to multivatiate cases is by computing the sum of local variables or by extracting the maximum value among the local CUSUM variables.

The SumCUSUM \cite{Mei2010} associates local CUSUM variables as follows:
\begin{equation}
    T_{SC}(n) = \frac{1}{L} \sum_{l=1}^L W_{l,n}.
    \label{eq11}
\end{equation}
Where $T_{SC}(n)$ being the global SumCUSUM variable; \textit{i.e.} the sum of the local CUSUM variables of the $L$ data stream. $W_{l,n}$ being the CUSUM variable of the $l^{th}$ sensor at time $n$. This variable $T_{SC}(n)$ will be compared to an adapted threshold $h$ to make a decision of a detection when $T_{SC}(n) > h$.

The MaxCUSUM extracts the highest value among the local CUSUM values as shown by Equation (\ref{eq12}): 
\begin{equation}
    T_{MC}(n) = \max_{l\in[1,\cdot\cdot\cdot,L]} (W_{l,n}).
    \label{eq12}
\end{equation}
It appears that SumCUSUM will be relevant to be used when all or almost all of the data streams are affected by the signal while MaxCUSUM will be relevant when one or only a few of the data streams are affected.

\cite{Mei2010} proposed to select ("censor") sensors and compute a partial and optimized SumCUSUM with a low computational cost. It seems to be a very effective way to merge the data for an online use of the method. The SumCUSUM variable is thus transformed as:
 \vspace{-0.2cm}
\begin{equation}
    T_{cSC}(n) = \frac{1}{\sum{\mathds{1}_{W_{l,n} \geq c}}}\sum_l^L W_{l,n}~\mathds{1}_{W_{l,n} \geq c} .
    \label{eq13}
\end{equation}
with $c$ a threshold based on the prior rough knowledge of the value $W_{l,n}$ would take if it were affected by the signal.
Figure \ref{fig6} shows the results of the three methods depending on the proportion of sensors affected by the signal.  The average run-length to false alarm (ARL2FA) of all three methods have been set to 30 000.  
We can see that when 1 or 2 out of 10 sensors are affected, the MaxCUSUM shows lower detection delays. \red{When 3 or more sensors are affected, the SumCUSUM gives a quicker detection than the MaxCUSUM.} We can also infer from Figure \ref{fig6} that the censored SumCUSUM is a good compromise between SumCUSUM and MaxCUSUM. However, \cite{Mei2010} shows that the best choice for $c$ depends on the number of sensors affected. While in some cases the proportion of sensors affected can be roughly predicted, in most cases it is completely unknown. In the case of the example of Figure \ref{fig6} $c$ has been set as 60\% of the global threshold $h$. 
\vspace{-0.4cm}
\begin{figure}[H]
\centering
\hspace{0cm}
    \includegraphics[width=0.9\hsize]{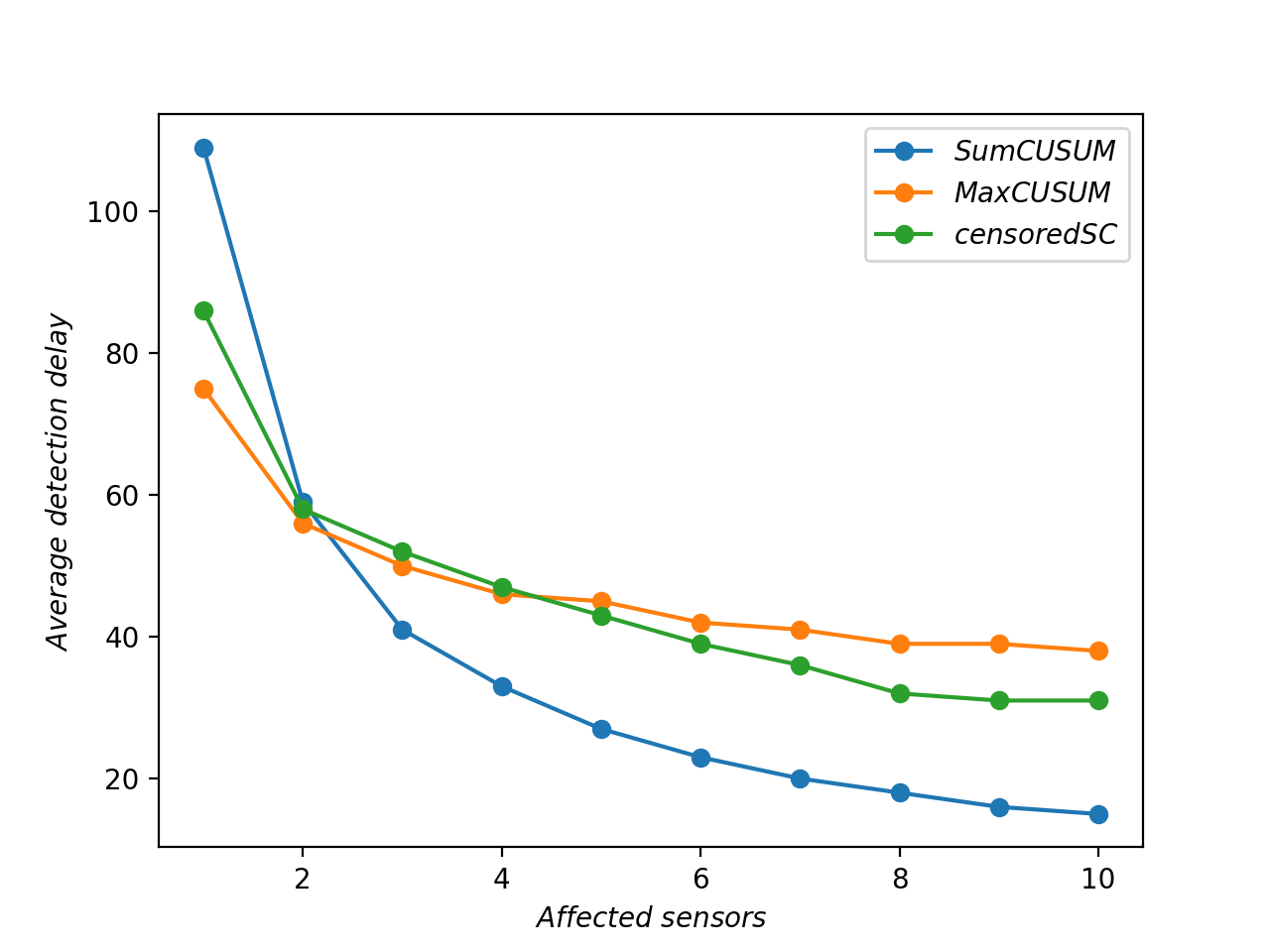}
    \caption{Average detection delay of SumCUSUM, MaxCUSUM and censored SumCUSUM techniques on a change in the mean of a Gaussian distribution with signal to noise ratio of -6dB. Change-point occurs at time sample 1000.}
    \label{fig6}
\end{figure}

\subsection{A novel adaptive censoring technique}

To overcome the limitation of requiring some prior knowledge on the expected values of $W_{l,n}$ in order to carefully choose the absolute threshold $c$ from \cite{Mei2010}, we propose a relative threshold $c_n$, computed for every time sample by:
\begin{equation}
    c_n = \alpha \times \max_{l \in [1, \cdot\cdot\cdot,L]}(W_{l,n})
    \label{eq_relativethreshold}
\end{equation}
with $\alpha$ being a factor so that $0\leq\alpha\leq 1 $.

In both cases the censoring technique is a compromise between the SumCUSUM and the MaxCUSUM. The results given by the two can be retrieved using particular values for $c$ ($0$ and $h$) or $\alpha$ ($0$ and $1$).

In order to assess the difference of behaviour of the two threshold types regarding the number of affected sensors when it is unknown, we conducted an experiment which results are shown in Figure \ref{fig10}. In this experiment, $h$ is set such that an average run-length to false alarm (ARL2FA) of 10,000 is obtained and $c$ and $\alpha$ values are set to be those which give the overall quickest detection for an unknown number of affected sensors between 1 and 20. The results empirically show that the proposed adaptive censoring technique outperforms the classical one. The gap in performance increases with the number of affected sensors.

\begin{figure}[H]
\hspace{0cm}
\centering
    \includegraphics[width=.9\hsize]{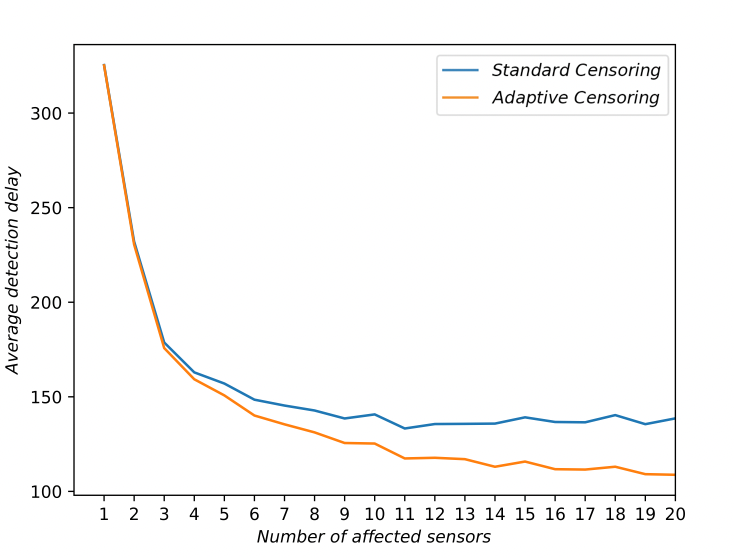}
    \caption{Average detection delay standard and adaptive Censored SumCUSUM at -12dB (with optimized values: $c = 5.8$ and $\alpha = 0.6$). A change of mean in the Gaussian distributions of a subset (in abscissa) of sensors appears at time 1000 while the other sensors keep the centered Gaussian distribution.}
    \label{fig10}
\end{figure}

 Because we cannot know in advance the number of sensors that will be affected in addition of the non-requirement of the knowledge of the expected values of $W_{l,n}$, the relative threshold is consequently more relevant.

Some clues can be pointed out to explain this difference of behaviour by examining the differences between the two methods in some particular cases. 

Case 1: All the local $W_{l,n} > c$ have close values one to another and are relatively far from the $W_{l,n} < c$. In this case both methods will compute the same $T_{cSC}$.

Case 2: The local $W_{l,n} > c$ have very different values. The standard method computes $T_{cSC}$ adding more low values of $W_{l,n}$ and has a lower value of $T_{cSC}$ which slows the detection.

Case 3: All the local $W_{l,n} > c$ have close values one to another and the $W_{l,n} < c$ are also close to the others. In this case, it is the standard method that computes a highest value for $T_{cSC}$, but the case implies that all the values are close to $c$ so detection does not happen in both cases unless the value chosen by $c$ is close to $h$ and in this case we have a behaviour close to the MaxCUSUM.

In all that follows we apply this optimised relative censoring technique to all local statistics (CUSUM, TC-CUSUM, FMA) and keep the SumCUSUM and the MaxCUSUM as benchmarks.

\subsection{Asynchronous monitoring and Temporary-Event-CUSUM on multivariate cases}

In the previous section we have considered that the signal appears simultaneously on all the affected sensors. Indeed, all the local test variables are computed simultaneously and it is from these that we can compute the global variable at time $n$ and make a decision regarding the detection.

In many practical cases the signal can be monitored by the sensors with a delay. Even more, some sensors can cease to be affected by the signal before some others begin to be. Thus, the sensors are not affected at the same time. One could say that we should try to find the best synchronicity of the data streams, meaning the synchronicity which maximises the associated CUSUM variable but this is a combinatorial problem.

By using locally the novel TE-CUSUM test statistic defined in \red{Proposition} \ref{prop_TransSigCUSUM}, the Sum-TE-CUSUM allows us to get the best synchronicity without requiring to test all the combinations, and thus saving a lot of computational resources. 
This time the global test variable becomes:
\begin{equation}
    T_{STEC}(n) = \frac{1}{L}\sum_{l=1}^L\max_{[\nu_l,N_l]}\sum_{k=\nu_l}^{N_l \leq n} \log \frac{f_{1,l}(X_{k,l}|\theta_{1,l})}{f_{0,l}(X_{k,l}|\theta_{0,l})}
    \label{eq14_1}
\end{equation}
$\nu_l$ is the change-point for the $l^{th}$ data stream and $N_l$ is the end of the signal presence in the $l^{th}$ data stream.

As a reminder of Equation (\ref{eq15}), the local variable is:
\begin{equation}
    G_n = \max(G_{n-1},W_n) = \max_{0<k\leq n}(W_k)
\end{equation}

From Equation (\ref{eq14}) and with $(\nu, N)$ of the $l^{th}$ data stream being rewritten $(\nu_l,N_l)$ :
\begin{equation}
    T_{STEC}(n) = \frac{1}{L}\sum_{l=1}^L G_{n,l}
    \label{eq19}
\end{equation}

The censoring technique can be applied to Equation (\ref{eq19}) simply by adding a threshold like in Equation (\ref{eq13}):

\begin{equation}
    T_{cSTEC}(n) = \frac{1}{\sum{\mathds{1}_{G_{n,l} > c_n}}}\sum_{l=1}^L G_{n,l} \times \mathds{1}_{G_{n,l} > c_n}
\end{equation}

Here is an example to illustrate the Sum-TE-CUSUM ($T_{STEC}$). An event is monitored in three data stream but with such a delay that there is no overlap. In Figure \ref{fig8}, we can see the three data streams with and without noise.

\vspace{-.35cm}
\begin{figure}[H]
\hspace{0cm}
\centering
    \includegraphics[width=0.9\hsize,height = 9cm]{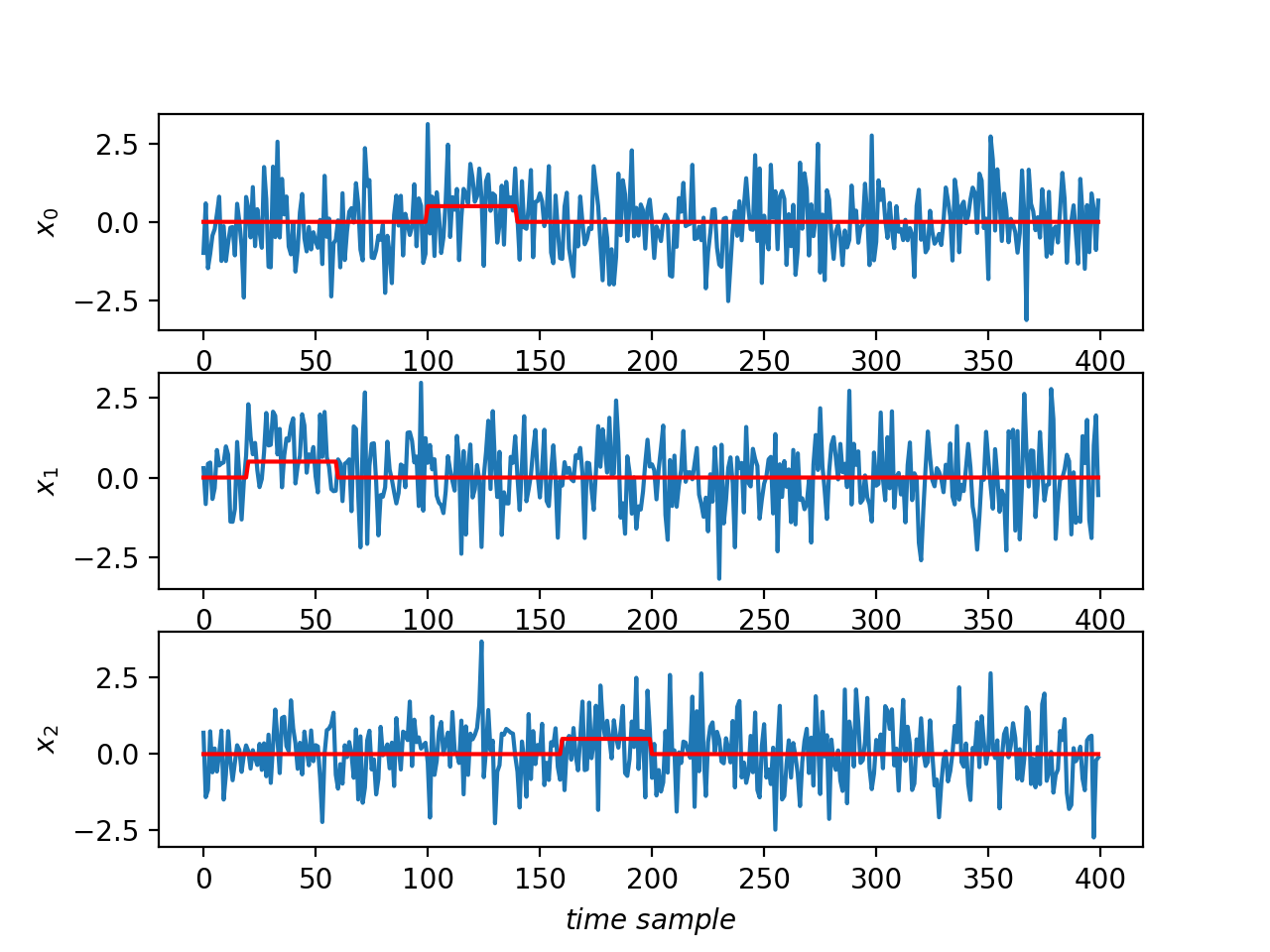}
    \caption{Three data streams monitoring a short event without overlapping (noiseless in red, noise+signal in blue)}
    \label{fig8}
\end{figure} 
\vspace{-0.5cm}
Figure \ref{fig9} shows the evolution of the test variable of the standard SumCUSUM technique and the TE-CUSUM.
\vspace{-0.5cm}
\begin{figure}[H]
\hspace{0cm}
\centering
    \includegraphics[width=0.9\hsize,height = 9cm]{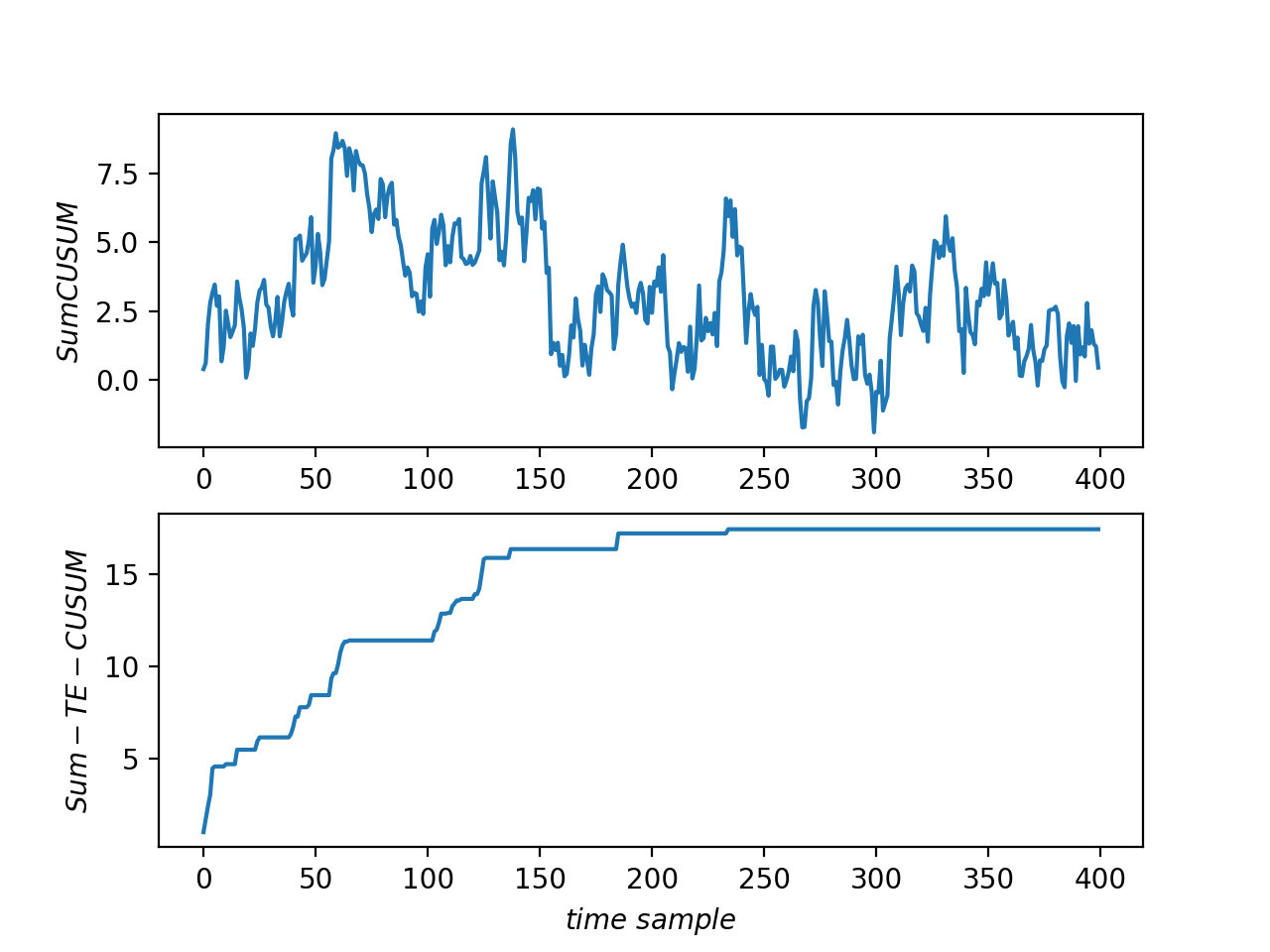}
    \caption{Evolution of test variables on the data streams of Figure \ref{fig8}}
    \label{fig9}
\end{figure} 

We can see in Figure \ref{fig9} that the SumCUSUM decreases between each appearance of the signal while the TE-CUSUM stands by and increases again as soon as the signal appears on another data stream. With the TE-CUSUM, we can detect the presence of the event with a higher threshold. If we set both thresholds in order to have a probability of false alarm of 1\% on this interval and if we make 10,000 runs we obtain a detection rate of 41\% with the SumCUSUM and 83\% with TE-CUSUM. \red{To achieve 1\% of false alarm, detection thresholds are set to 17.5 for the TE-CUSUM and to 13 for the SumCUSUM.}

A pseudo-code, explaining the TE-CUSUM function is available in \ref{appendix3}.

Remark: In order to compare it to the other methods, we can also use the censoring technique to extend the FMA technique to multivariate cases:

\begin{equation}
    T_{cFMA}(n) = \frac{1}{\sum_l \mathds{1}_{Z_{n,l}^w>c_n}}\sum_{l=1}^L Z_{n,l}^w \times  \mathds{1}_{ Z_{n,l}^w > c_n}
\end{equation}

\section{Validation}

In this section we compare the different detection methods presented previously. The studied methods are used to detect a change in the mean amplitude $A$ affecting only a subset of sensors. The measurement noise at each sensor is assumed to be normally distributed with zero mean and standard deviation $\sigma$. This experiment is conducted on several cases with two different signal to noise ratios (SNR) defined as $SNR = 10 \log_{10}(\frac{A^2}{\sigma^2 })$. \red{We ensured to be out of transitional mode in the signal by exposing the first affected sensor after 1500 time samples. In the following plots for simplicity, zero corresponds to the beginning of exposure and therefore not the beginning of the experiment. }
Ten sensors are considered among which 3 or 7 monitor the event. \red{The case where 5 sensors monitor the signal is displayed in \ref{sec:sample:appendix}.} The censoring technique will be applied to TE-CUSUM and FMA. For each method the global threshold $h$ is set to have an ARL2FA of 30,000 time samples, and the $\alpha$ parameter from the censoring technique is set (except for SumCUSUM and MaxCUSUM) for the quickest detection for a random number of signal but when all of them are synchronized. \red{In order to set the ARL2FA to 30,000 time samples, the detection thresholds for each method have been set so we have the cumulative probabilities of false alarm presented Figure \ref{fig_calib}}

\begin{figure}
\vspace{-3cm}
\centering
    \includegraphics[width=0.9\hsize,height = 7.5cm]{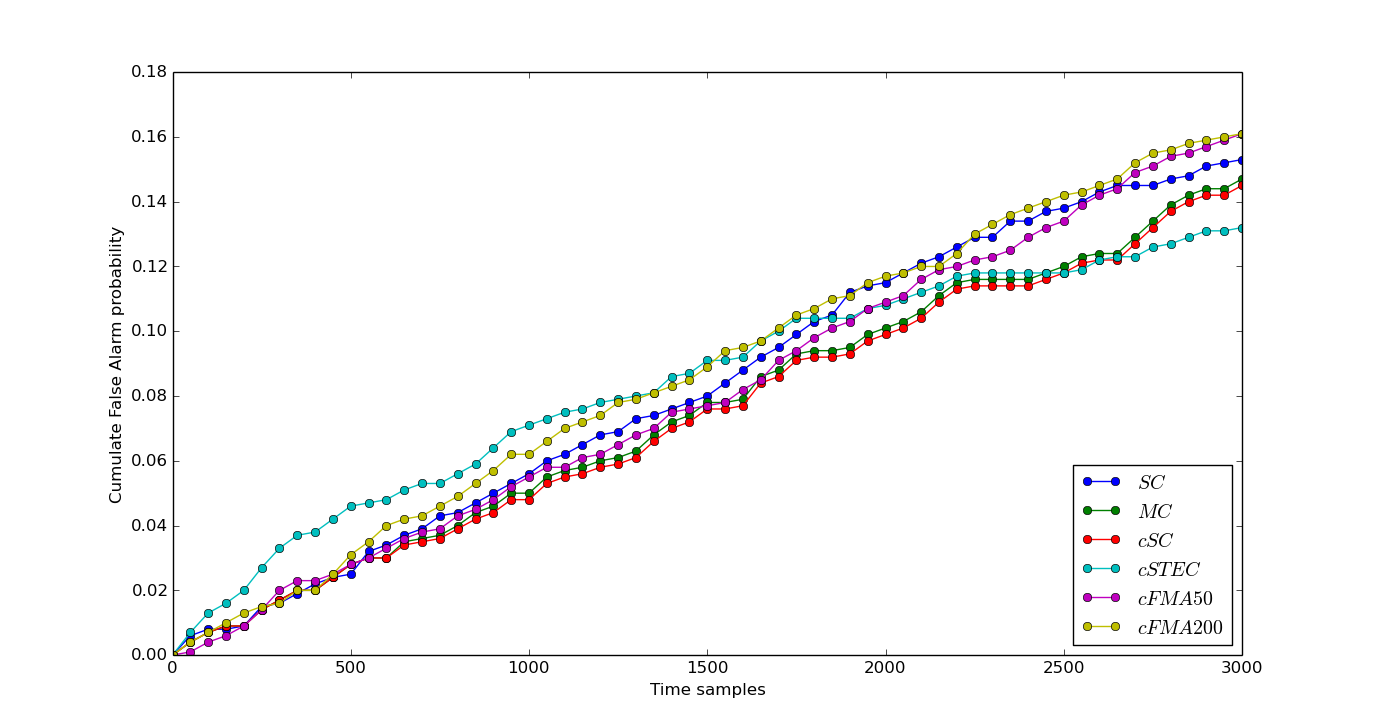}
    \caption{Empirical cumulative false alarm probability for the experiment presented in the validation section. All the methods have been set to have close values of false alarm probability. From this figure we can deduce that the probability of a false alarm to happen in a 1000 time samples window is approximately of 5\%. This means that there is a 5\% chance of false alarm before the beginning of the signal, 5\% chance that a false alarm leads to detection in the fist scenario and 10\% chance that a false alarm leads to detection in the second scenario.(SC=SumCUSUM, MC=MaxCUSUM, cSC = Censored-SumCUSUM, cSTEC = Censored-Sum-Temporary-Event-CUSUM, cFMA50 = Censored-Finite-Moving-Average with a window of 50 time samples and cFMA200 = Censored-Finite-Moving-Average with a window of 200 time samples).}
    \label{fig_calib}
\end{figure} 

Two sizes of window have been chosen for the FMA technique. One window with 50 time samples and another with 200 time samples. Because the duration of the event is supposed to be unknown, this will show a case where the time window is longer than the exposure and a case where the time window shorter than the exposure. In the scenario where the SNR is $-8dB$. In the second scenario, where the SNR is $-14dB$, the duration of the exposure will, fortunately for the FMA200, be of 200 time samples. 

The numerical experiments are divided into four cases: the first case is when all the signals are monitored simultaneously, in the second case there is a drift of half the signal length between each sensor which monitors the signal, in the third case the drift is of a full signal length \lch{(This case is displayed in appendix \ref{sec:sample:appendix})} and in the fourth the drift is of one and a half signal length. This last simulation tends to represent a diffuse event passing through a field of sensors so that these ones do not monitor the event at the same time and only a little part of the sensors can monitor the signal anyways.  

\subsection{First scenario}

In this first scenario, all the studied techniques are set to expect an offset of 0.4 in amplitude and to have an ARL2FA of 30,000 time samples. When a sensor is affected by the signal, its mean value is affected by an offset of 0.4 for a duration of 100 time samples. \red{We expect the detection ratio of the methods to grow rapidly and then become steady shortly after the end of exposure as no signal is added to the system anymore. The growth of each slope will indicate how fast the detection is, and the value reached after the end of exposure indicates the chances of detection given the current experiment.}

\begin{figure}[H]
\begin{subfigure}{.5\textwidth}
\hspace{0cm}
    \includegraphics[width=1\hsize]{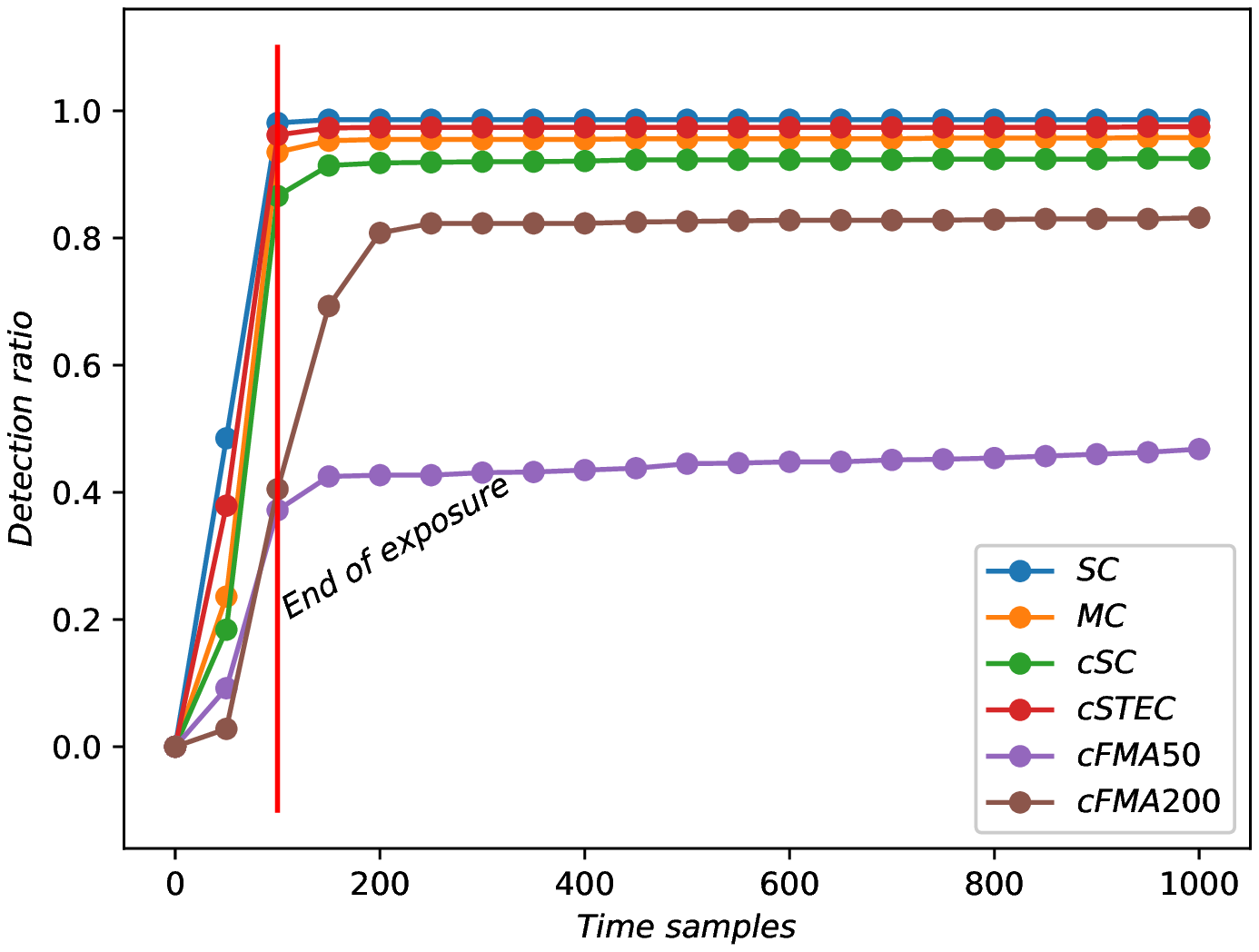}
    \caption{Results when 3 sensors out of 10 are affected}
    \label{fig_Sync3_04}

\end{subfigure}%
\begin{subfigure}{.5\textwidth}

\hspace{0cm}
    \includegraphics[width=1\hsize]{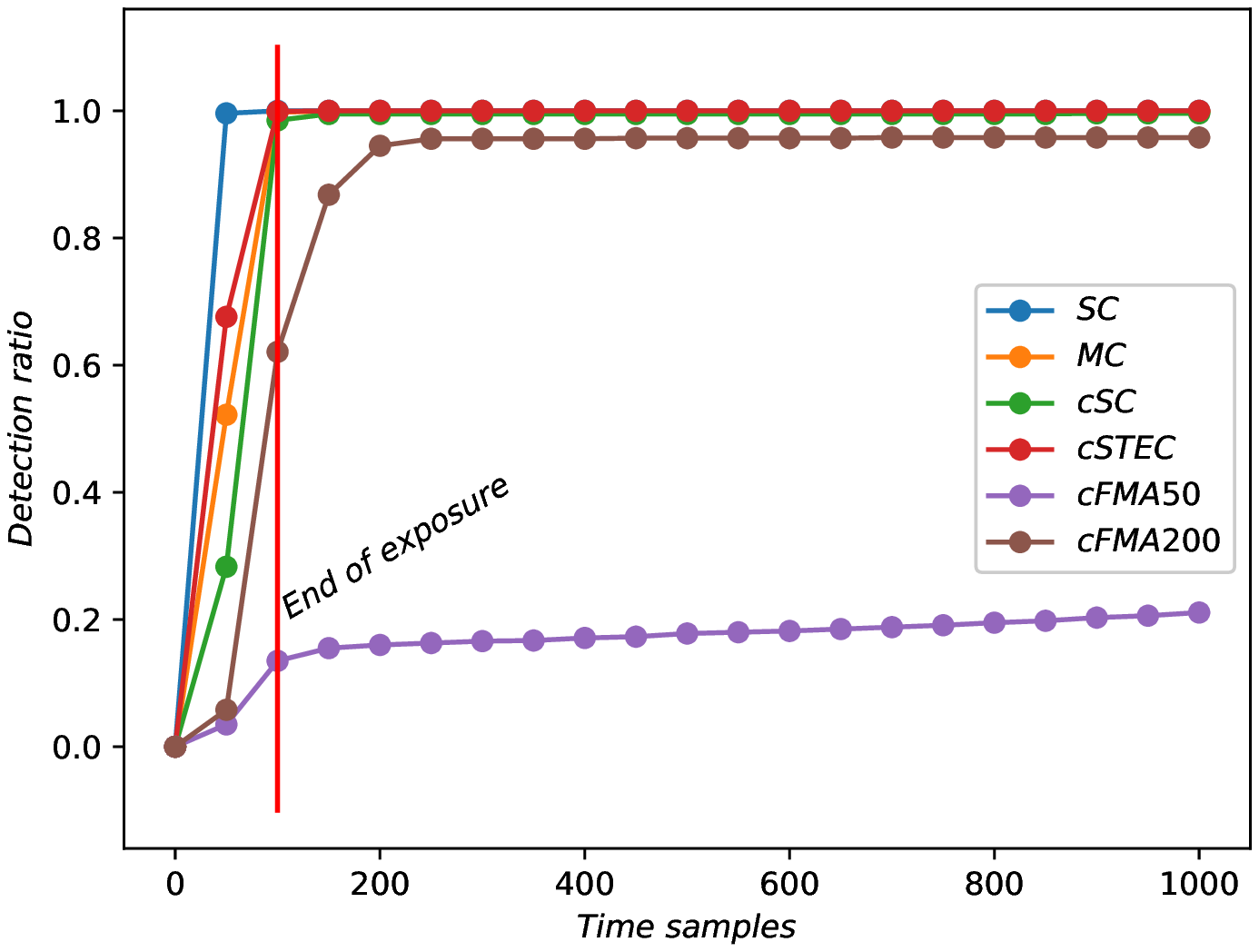}
    \caption{Results when 7 sensors out of 10 are affected}
    \label{fig_Sync7_04}
\end{subfigure}
\caption{Cumulative detection rate as a function of time of each of the presented method when streams are synchronized with a signal to noise ratio of -8dB. The end of exposure marks the time until at least one sensor is monitoring the change in amplitude. The exposure time is of 100 time samples for the affected sensors.(Legend is detailed in Figure \ref{fig_calib}).}
\end{figure}

From Figures \ref{fig_Sync3_04} and \ref{fig_Sync7_04}, we can see that when all the exposures are synchronized, the SumCUSUM technique is the one that gives the best results. However, the censored-SUM-TE-CUSUM and the Max-CUSUM give rather good results in this case. When data streams stop to monitor the event, at the time marked "end of exposure", the detection ratio is over 80\% for the four best methods in the case where 3 data sensors are exposed and over 90\% when 7 sensors are exposed. In the case 7 sensors are exposed the censored TE-CUSUM as well as the SumCUSUM give a result of almost 100\% of detection at the end of exposure. The FMA technique seems to give slower detection and does not manage to reach the detection rate of the other techniques.

\begin{figure}[H]

\begin{subfigure}{.5\textwidth}
\hspace{0cm}
    \includegraphics[width=1\hsize]{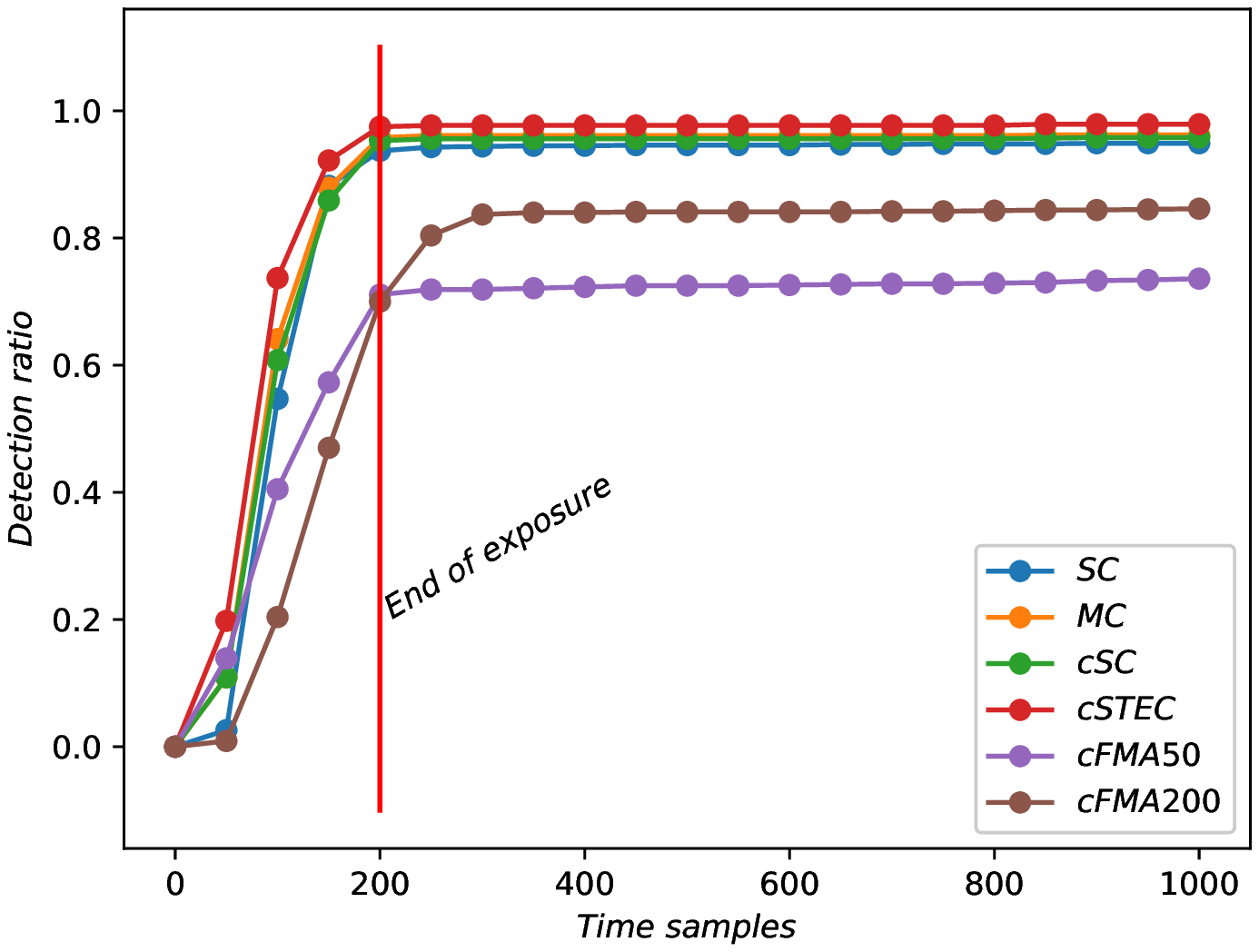}
    \caption{Results when 3 sensors out of 10 are affected}
    \label{fig_OSync3_04}
\end{subfigure}%
\begin{subfigure}{.5\textwidth}
\hspace{0cm}
    \includegraphics[width=1\hsize]{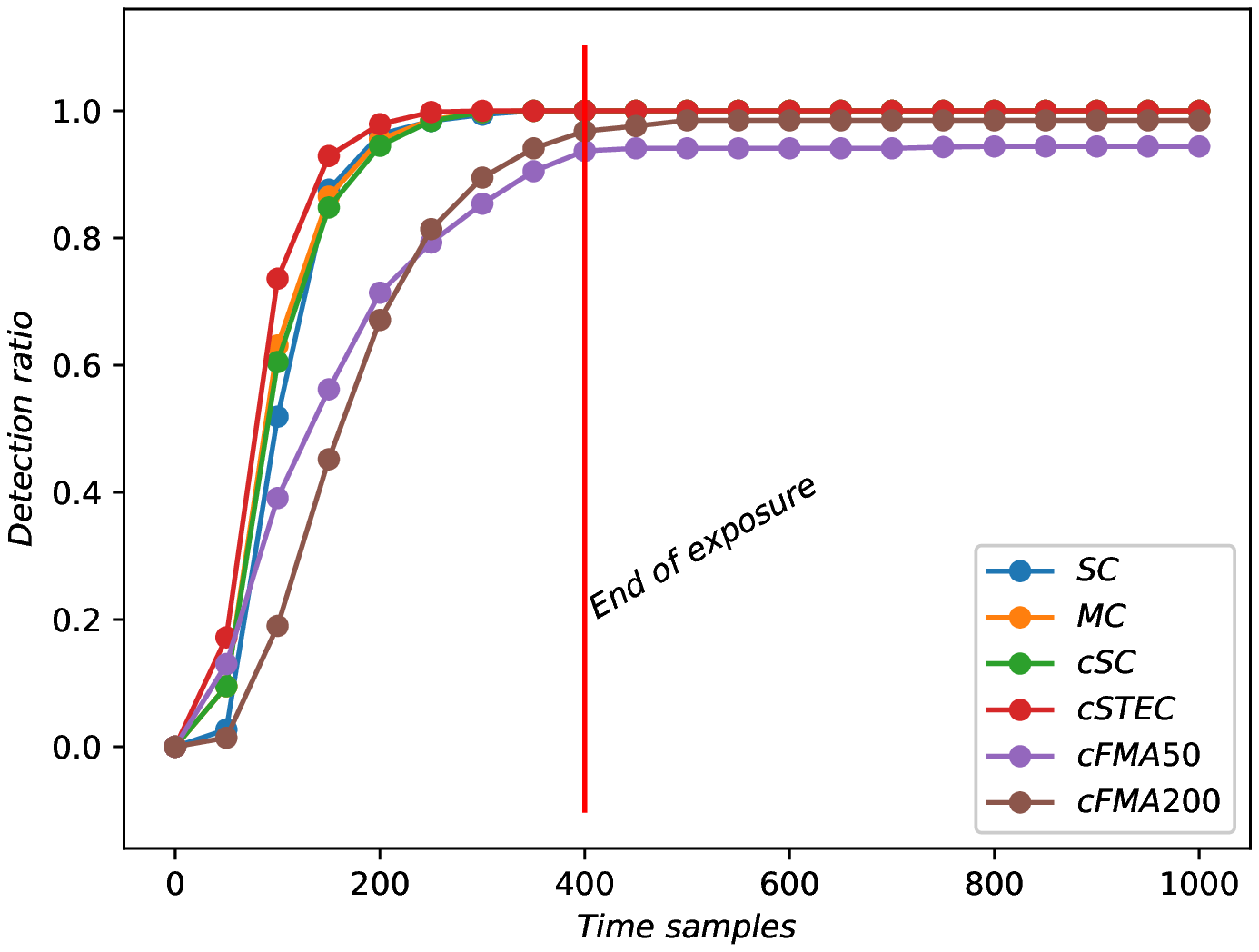}
    \caption{Results when 7 sensors out of 10 are affected}
    \label{fig_OSync7_04}
    \end{subfigure}%
    \caption{Cumulative detection rate as a function of time of each of the presented method when streams are slightly out-of-sync with a signal to noise ratio of -8dB. The end of exposure marks the time until at least one sensor is monitoring the change in amplitude. Here the delay between the beginning of the exposure of a sensor and the next is of 50 time samples while the exposure time by sensor is of 100 time samples.(Legend is detailed in Figure \ref{fig_calib}).}
\end{figure}

The results presented on Figures \ref{fig_OSync3_04} and \ref{fig_OSync7_04} are obtained in the same conditions except this time, the start of monitoring of every data stream is delayed by 50 time samples from the previous one.
Here, the censored TE-CUSUM shows how it manages to give similar results in more difficult conditions. Indeed, at the end of the exposure in figure \ref{fig_OSync3_04} the total energy transmitted by the event to the system is the same than at the end of exposure in figure \ref{fig_Sync3_04}. The censored TE-CUSUM gives in both cases a 90\% detection rate at the "end of exposure" where the SumCUSUM technique falls from about 95\% when signals are synchronized to a little less than 90\% when they are slightly out-of-sync. On these figures we can also see that the longer the overall exposure is the better the FMA results seem to be.

\begin{figure}[H]

\begin{subfigure}{.5\textwidth}
\hspace{0cm}
    \includegraphics[width=1\hsize]{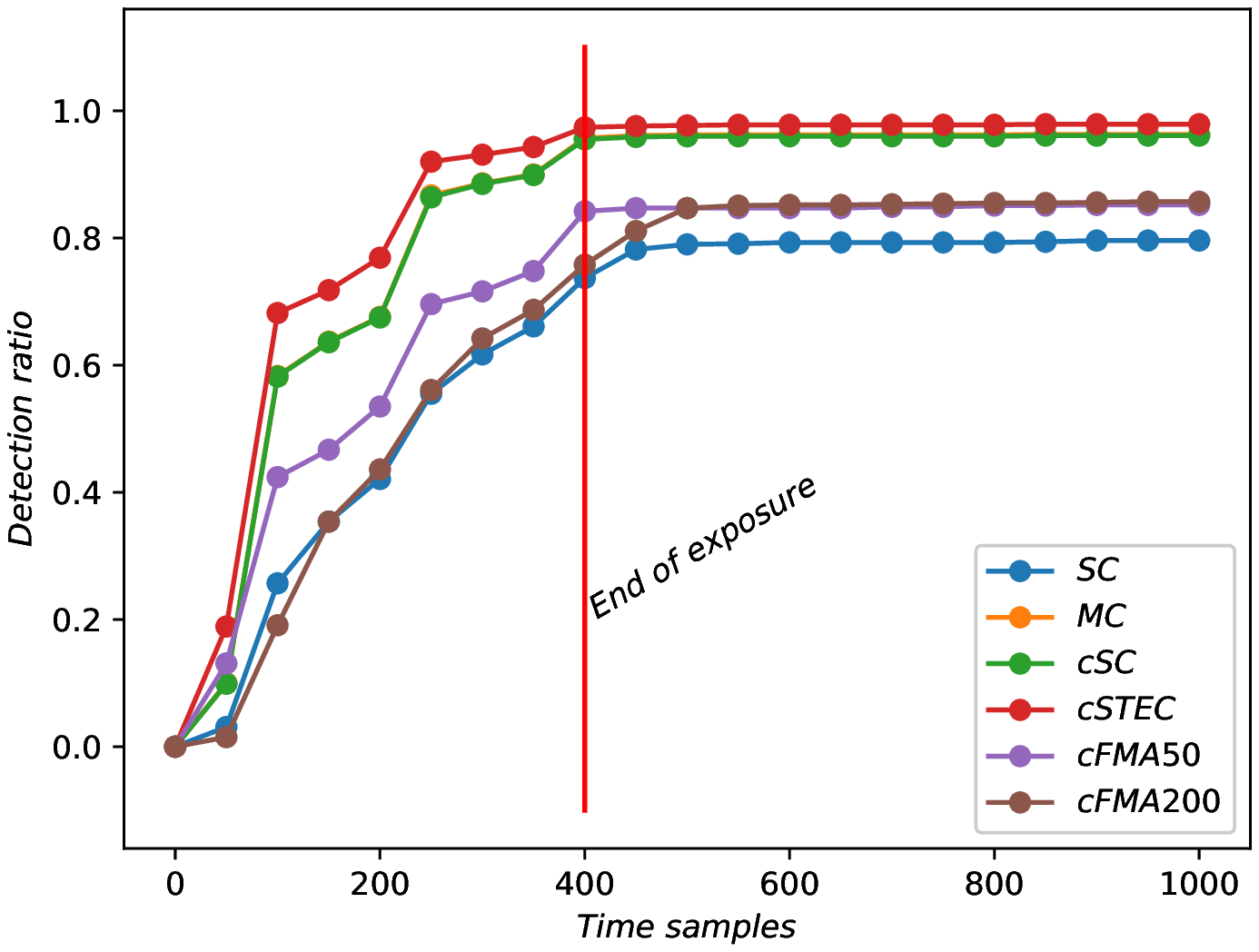}
    \caption{Results when 3 sensors out of 10 are affected}
    \label{fig_OOOSync3_04}
\end{subfigure}%
\begin{subfigure}{.5\textwidth}
\hspace{0cm}
    \includegraphics[width=1\hsize]{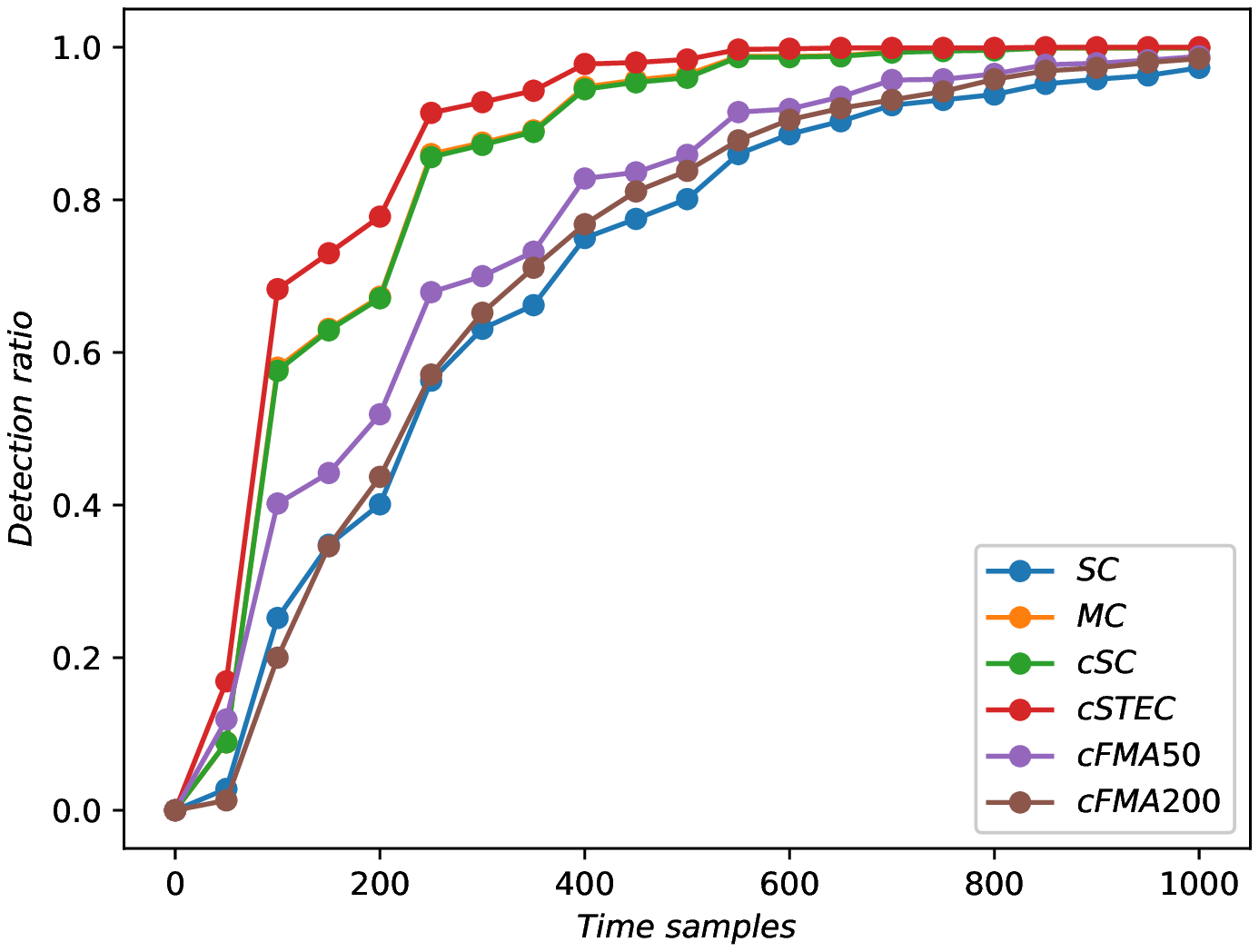}
    \caption{Results when 7 sensors out of 10 are affected }
    \label{fig_OOOSync7_04}
    \end{subfigure}%
\caption{Cumulative detection rate as a function of time of each of the presented method when streams are totally out-of-sync with a signal to noise ratio of -8dB. The end of exposure marks the time until at least one sensor is monitoring the change in amplitude. In (b) the end of exposure happen after the end of monitoring. Here the delay between the beginning of the exposure of a sensor and the next is of 150 time samples while the exposure time by sensor is of 100 time samples. (Legend is detailed in Figure \ref{fig_calib}).}
\end{figure}

In the case of Figures \ref{fig_OOOSync3_04} and \ref{fig_OOOSync7_04} there is a gap of 50 time samples where no data stream monitors the signal between the exposure of each data stream. One can see that the detection rate tends to stay still between the exposure of each data stream. In figure \ref{fig_OOOSync7_04} the end of exposure happens at time sample 1100 and is not marked on the figure.
In these two last cases we can see that the censored FMA technique, can give good results when the exposure lasts. It is notable that in all the non synchronized cases, censored TE-CUSUM gives the best early detection and the best detection rate.  This time, the SumCUSUM, because it only considers instantaneous CUSUM variables, has a very low detection rate compared to the other techniques. The difference of performances between the SumCUSUM  and its censored version clearly shows the benefit of using our proposed censoring technique when signals are not synchronised. A similar behaviour has been observed with the censoring technique applied to our TE-CUSUM.

Results with more cases are presented in \ref{sec:sample:appendix}.

\subsection{Second scenario}

In this second scenario, we run the same tests, but with a lower SNR. The amplitude of the signal expected by the methods is still of 0.4, but the real signal has an amplitude of only 0.2. Again, detection thresholds of all methods are set to have an ARL2FA of 30,000 time samples.  When a sensor is affected by the signal, its mean value is affected by an offset of 0.2 for a duration of 200 time samples.

\begin{figure}[H]

\begin{subfigure}{.5\textwidth}
\hspace{0cm}
    \includegraphics[width=1\hsize]{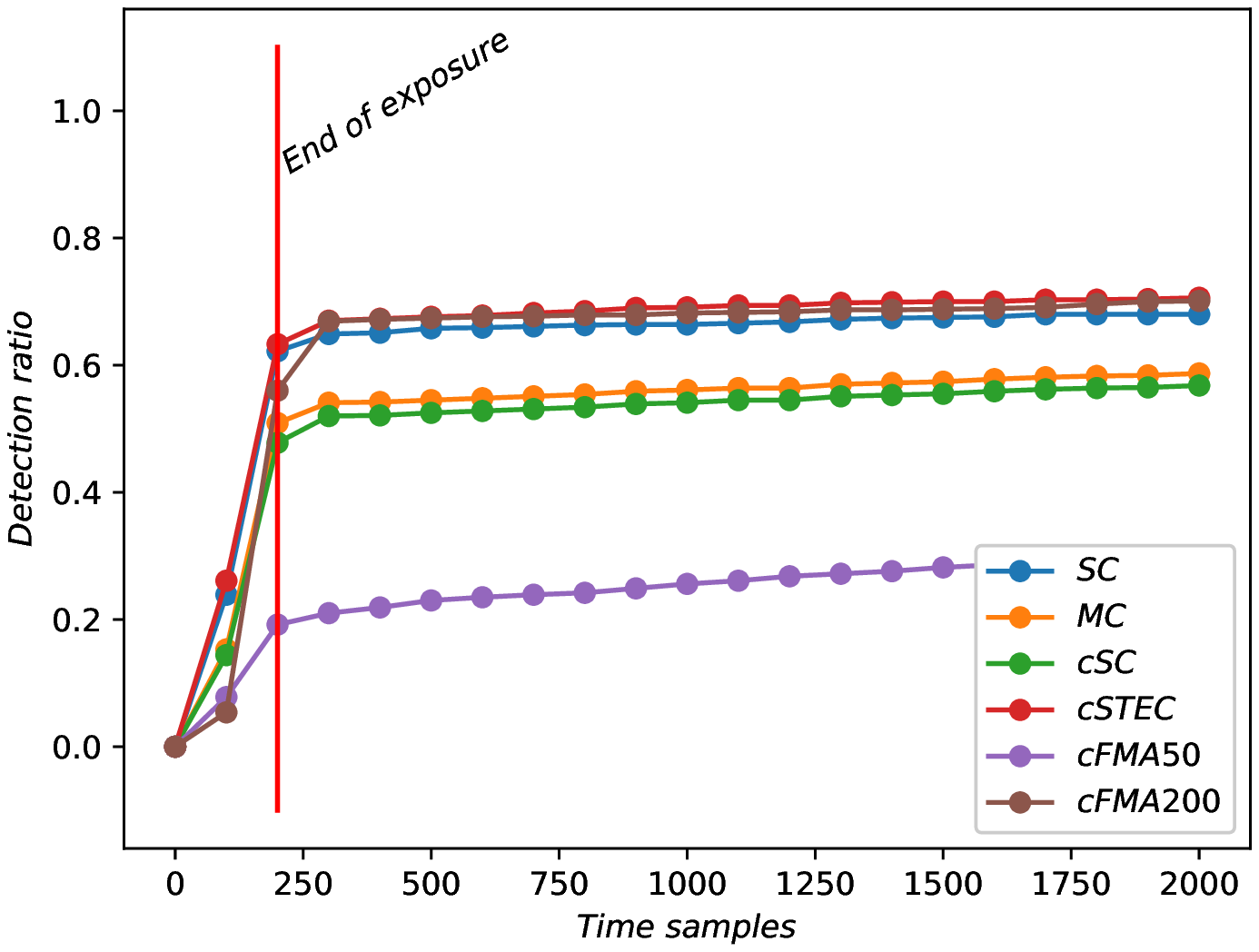}
    \caption{Results when 3 sensors out of 10 are affected}
    \label{fig_Sync3_02}
\end{subfigure}%
\begin{subfigure}{.5\textwidth}
\hspace{0cm}
    \includegraphics[width=1\hsize]{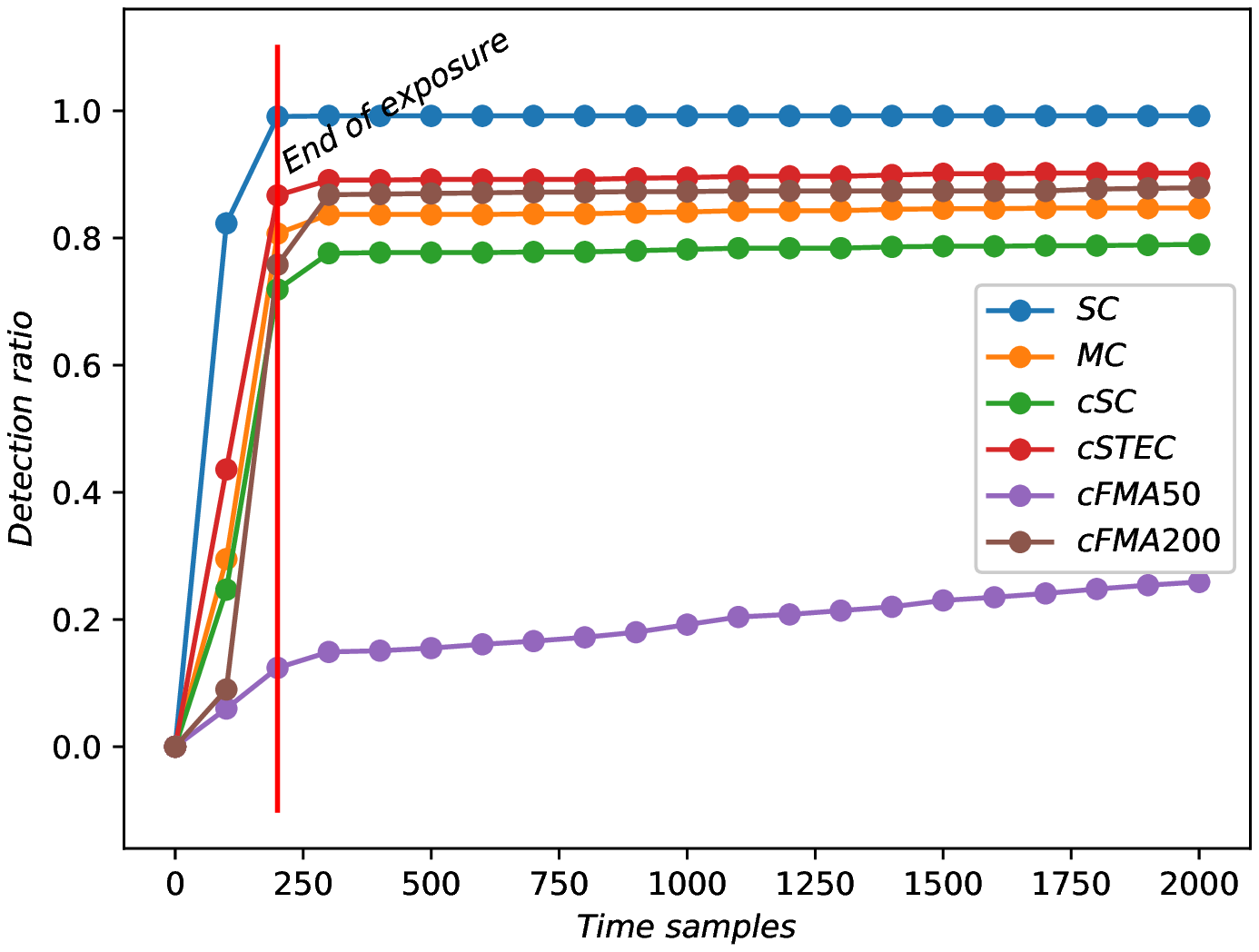}
    \caption{Results when 7 sensors out of 10 are affected}
    \label{fig_Sync7_02}
    \end{subfigure}%
    \caption{Cumulative detection rate as a function of time of each of the presented method when streams are synchronized with a signal to noise ratio of -14dB. The end of exposure marks the time until at least one sensor is monitoring the change in amplitude. The exposure time is of 100 time samples for the affected sensors. (Legend is detailed in Figure \ref{fig_calib}).}
\end{figure}
\vspace{-0.5cm}
Figures \ref{fig_Sync3_02} and \ref{fig_Sync7_02} shows that, this time, nearly all the methods have the utmost difficulties to give high detection rates. The exposure of more data streams in the case of Figure \ref{fig_Sync7_02} improves the detection rates of the methods, but SumCUSUM is the only one which is fully satisfying. 

\vspace{-0.5cm}
\begin{figure}[H]

\begin{subfigure}{.5\textwidth}
\hspace{0cm}
    \includegraphics[width=1\hsize]{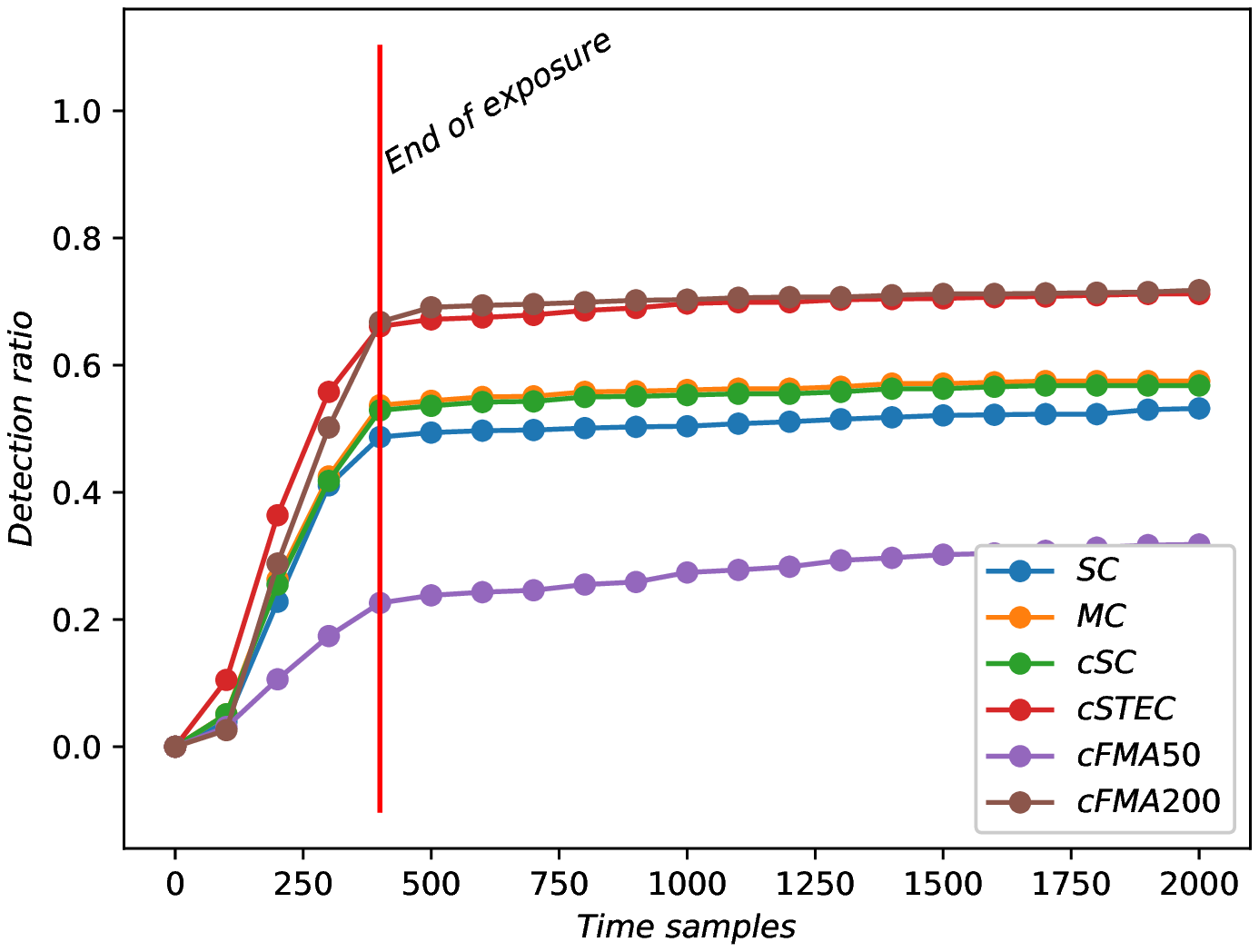}
    \caption{Results when 3 sensors out of 10 are affected}
    \label{fig_OSync3_02}
\end{subfigure}%
\begin{subfigure}{.5\textwidth}
\hspace{0cm}
    \includegraphics[width=1\hsize]{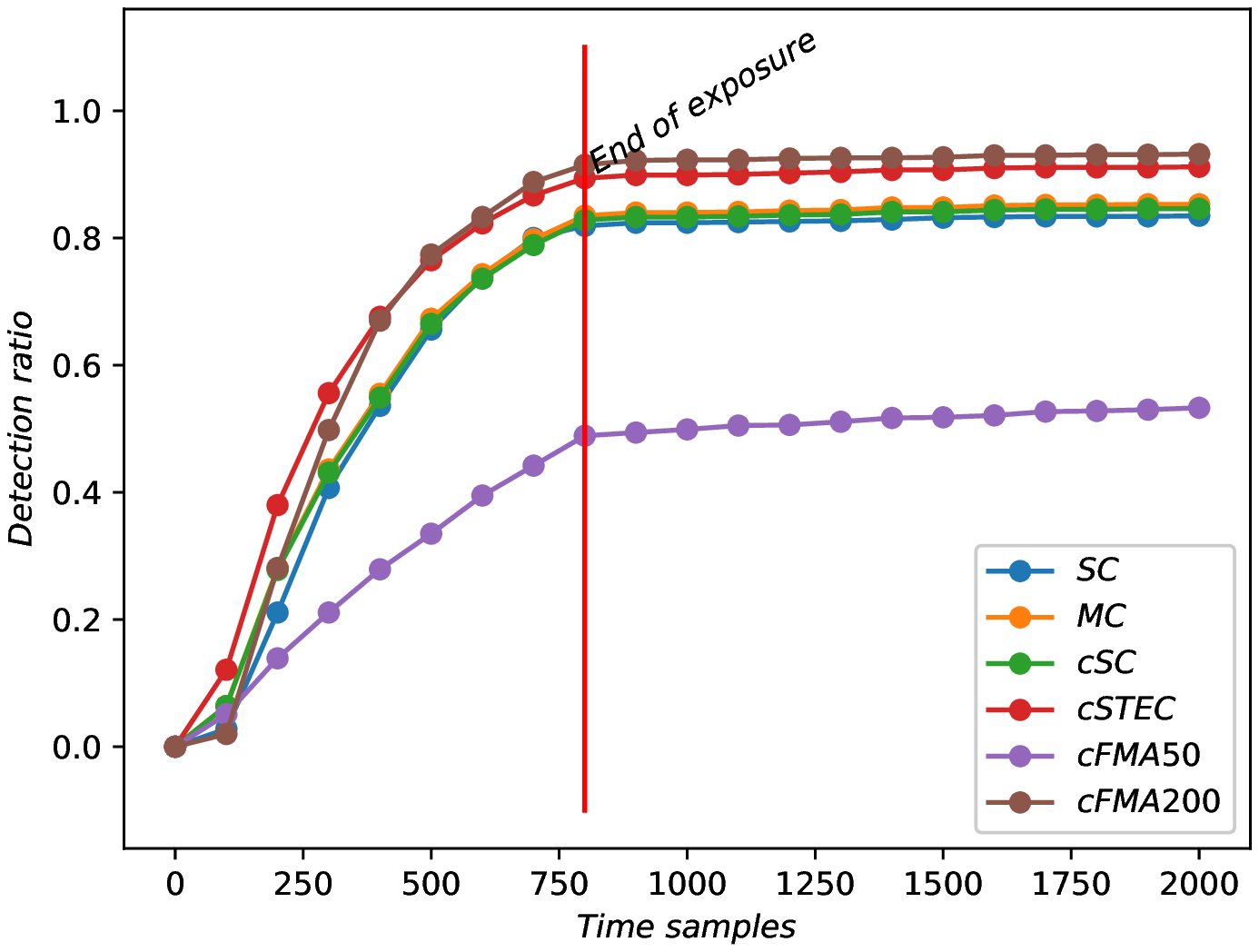}
    \caption{Results when 7 sensors out of 10 are affected}
    \label{fig_OSync7_02}
    \end{subfigure}%
    \caption{Cumulative detection rate as a function of time of each of the presented method when streams are slightly out-of-sync with a signal to noise ratio of -14dB. The end of exposure marks the time until at least one sensor is monitoring the change in amplitude. Here the delay between the beginning of the exposure of a sensor and the next is of 50 time samples while the exposure time by sensor is of 100 time samples. (Legend is detailed in Figure \ref{fig_calib}).}
\end{figure}

Figures \ref{fig_OSync3_02} and \ref{fig_OSync7_02} show that when signals are slightly out-of-sync, the censored FMA (with a time window of 200) and censored TE-CUSUM give similar results. If is these results are not fully satisfying, it is important to remember that the SNR is very low. In terms of power, the signal to noise ratio is of -14dB on a data stream when the signal is present. 

\begin{figure}[H]

\begin{subfigure}{.5\textwidth}
\hspace{0cm}
    \includegraphics[width=1\hsize]{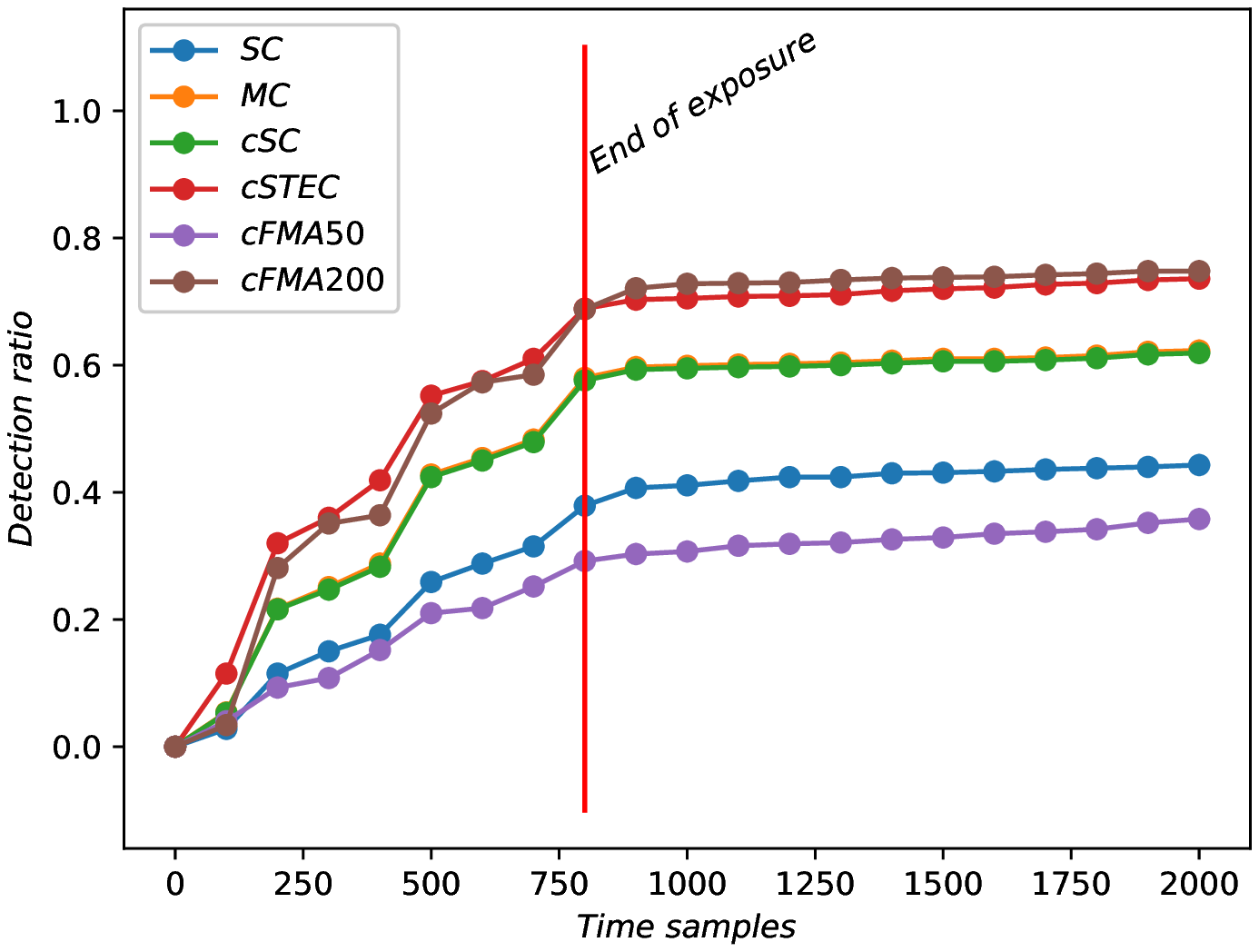}
    \caption{Results when 3 sensors out of 10 are affected}
    \label{fig_OOOSync3_02}
\end{subfigure}%
\begin{subfigure}{.5\textwidth}
\hspace{0cm}
    \includegraphics[width=1\hsize]{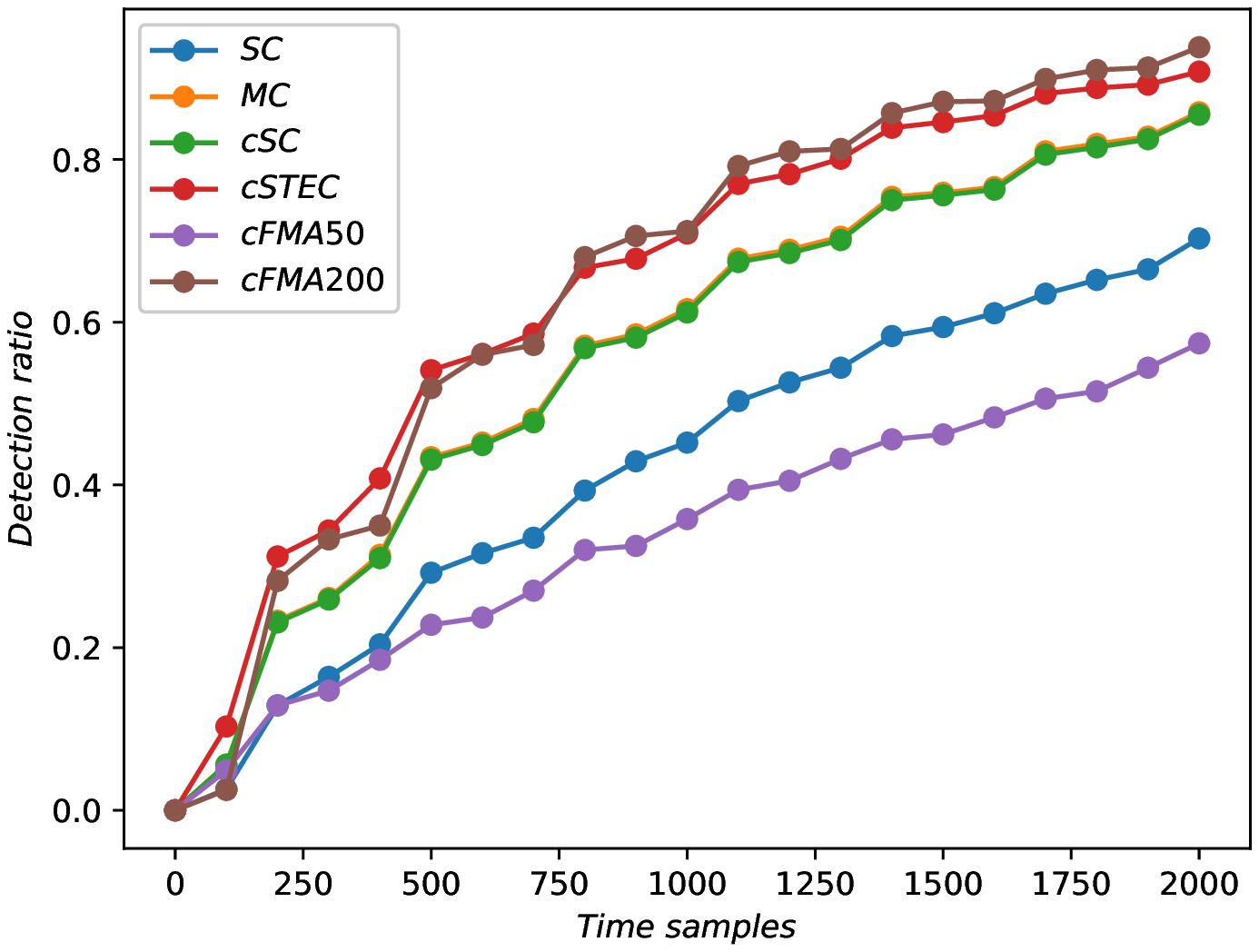}
    \caption{Results when 7 sensors out of 10 are affected}
    \label{fig_OOOSync7_02}
\end{subfigure}%
\caption{Cumulative detection rate as a function of time of each of the presented method when streams are totally out-of-sync with a signal to noise ratio of -14dB. The end of exposure marks the time until at least one sensor is monitoring the change in amplitude. In (b) the end of exposure happen after the end of monitoring. Here the delay between the beginning of the exposure of a sensor and the next is of 150 time samples while the exposure time by sensor is of 100 time samples. (Legend is detailed in Figure \ref{fig_calib}).}
\end{figure}

In this last example, Figures \ref{fig_OOOSync3_02} and \ref{fig_OOOSync7_02} show that the censored TE-CUSUM and the censored FMA can reach the same detection rate as when signals are synchronized or slightly out-of-sync, even if these detection rate are reached later than when signals are synchronised. The SumCUSUM gives very good results when most sensors monitor the event and when these are synchronized. However, when one of these condition is not met, the SumCUSUM performance decreases very fast.  

Results with subsets of affected sensors and one additional asynchronous case are presented in \ref{sec:sample:appendix}.

\section{Conclusion}

In this paper, we have addressed the detection problem of an event which only appears on a subset of \lch{sensors}, and such that these appearances can be delayed one to another so that they can be perceived by the system as if the data streams monitoring the same event were out-of-sync. 

Existing methods have already explored the fact that an event can be monitored by only a portion of the \lch{sensors}, as well as they can deal with the fact that the change point does not occur at the same time on every data stream. But what standard CUSUM methods lack to consider is the fact that a sensor can in some cases ceases to monitor the event while another one does.

The method we propose takes in consideration all the cases of delays between the change point on the different data streams as well as the fact that some of them can cease to monitor the event before the end of the system exposure. We have shown that if the system is composed of only one data stream, the TE-CUSUM is equivalent to a standard CUSUM procedure. 

We have also shown that the censored TE-CUSUM, beside the cases where all signals are synchronized, gives the best results, even if those can be similar to the FMAs at very low SNR when the FMA window is adapted to the signal length. It is important however to note that the FMA was not originally designed to be used on multivariate out-of-sync cases. Also, a big advantage of the censored TE-CUSUM, is that it keeps the recursive computation of the CUSUM. Indeed, it only adds a comparison of the SumCUSUM variable to the last maximum to the SumCUSUM technique where the FMA computes the likelihood ratio on a signal portion which can be rather long, which means it needs to store many values (\ref{appendix2}) to compute the test variable. The FMA also requires to use the window length which can be sensitive to the length of the signal that is expected. This sensitive parameter of the window length for the FMA gives the TE-CUSUM the advantage of being easier to tune when the length of exposure is not known.

In this paper we have introduced the new TE-CUSUM method which provides a light and simple detection technique which covers a greater range of cases by adding the possibility for the event to detect to be temporarily and not simultaneously monitored by the different data streams while still giving rather good results in standard cases when data streams are synchronised. The proposed procedure can have many practical applications, e.g. when a network of sensors is monitoring a localised event passing through. It can be, for instance a plume travelling into the air containing a chemical compound one wishes to detect or a furtive object passing through several radar monitored areas one after the other.

\newpage
\appendix
\section{TE-CUSUM pseudo-code}
\label{appendix3}

\begin{algorithmic}
\State{$Detection \gets False$}
\State {$t \gets 0$}
\For{$l \in L$}

    \State {$G_{l,0} \gets 0$}
    
    \State {$W_{l,0} \gets 0$}
    
\EndFor
\While{Detection = False}
    \State {$t \gets t+1$}
    
    \For{$l \in L$}
        \State $X_{l,t} \gets value~of~sensor~l~at~time~t$ 
        \State $L_{l,t} \gets log-likelihood~for~sensor~l$
        \If{$W_{l,t-1} \geq 0$}
            \State{$W_{l,t} \gets W_{l,t-1} + L_{l,t}$}
        \Else
            \State{$W_{l,t} \gets L_{l,t}$}
        \EndIf
        \If {$W_{l,t}> G_{l,t-1}$}
            \State{$G_{l,t} \gets W_{l,t}$}
        \Else
            \State{$G_{l,t} \gets G_{l,t-1}$}
        \EndIf
    \EndFor
    \State {$TECUSUM \gets 0$}
    \State {$s \gets 0$}
    \For{$l \in L$}
        \If{$G_{l,t} > \alpha \times \underset{l}{\max} G_{l,t}$}
        
            \State {$ TECUSUM \gets  TECUSUM + G_{l,t}$}
            \State {$s \gets s+1$}
        \EndIf
    \EndFor
    \If{$\frac{TECUSUM}{L_t} \geq threshold$}
        \State{$Detection \gets True$}
        
    \EndIf
    
\EndWhile
\end{algorithmic}
\newpage


\section{Experimental results}
\label{sec:sample:appendix}
\begin{table}[H]
\begin{tabular}{c c c}
    \includegraphics[width=0.33\hsize]{fig/Sync3_04.eps} & \includegraphics[width=0.33\hsize]{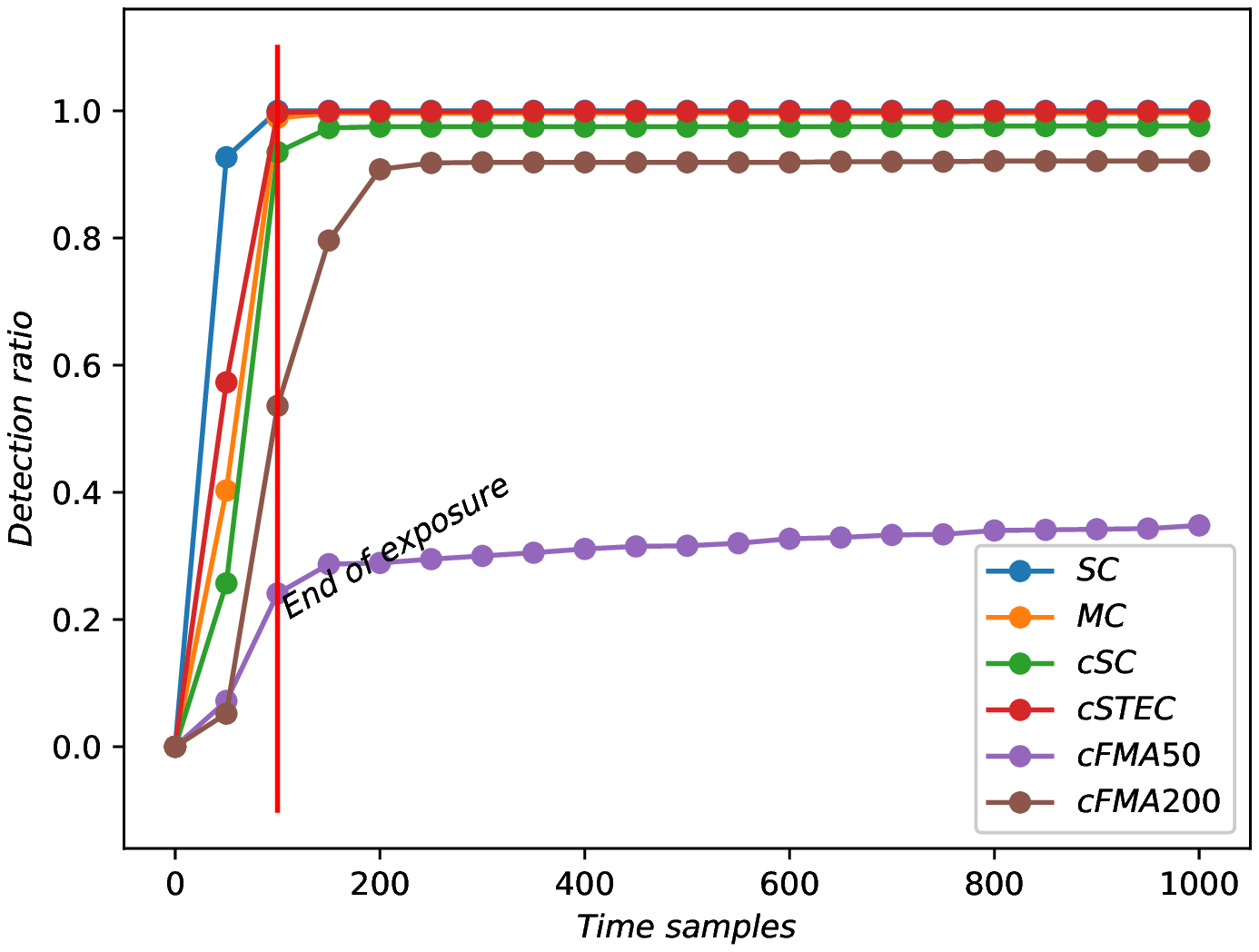}& 
    \includegraphics[width=0.33\hsize]{fig/Sync7_04.eps} \\
    \includegraphics[width=0.33\hsize]{fig/OSync3_04.eps} & \includegraphics[width=0.33\hsize]{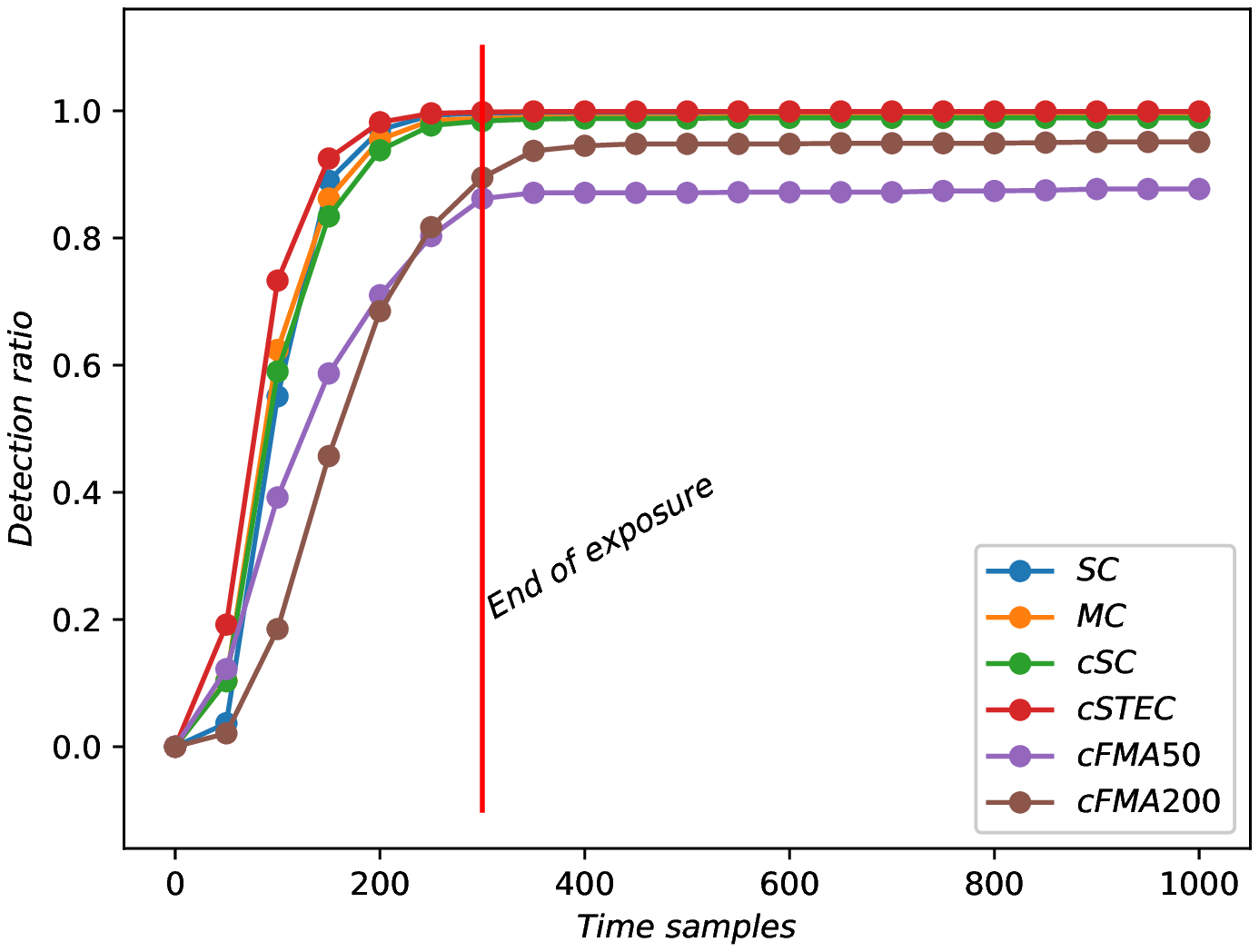}& 
    \includegraphics[width=0.33\hsize]{fig/OSync7_04.eps} \\
    \includegraphics[width=0.33\hsize]{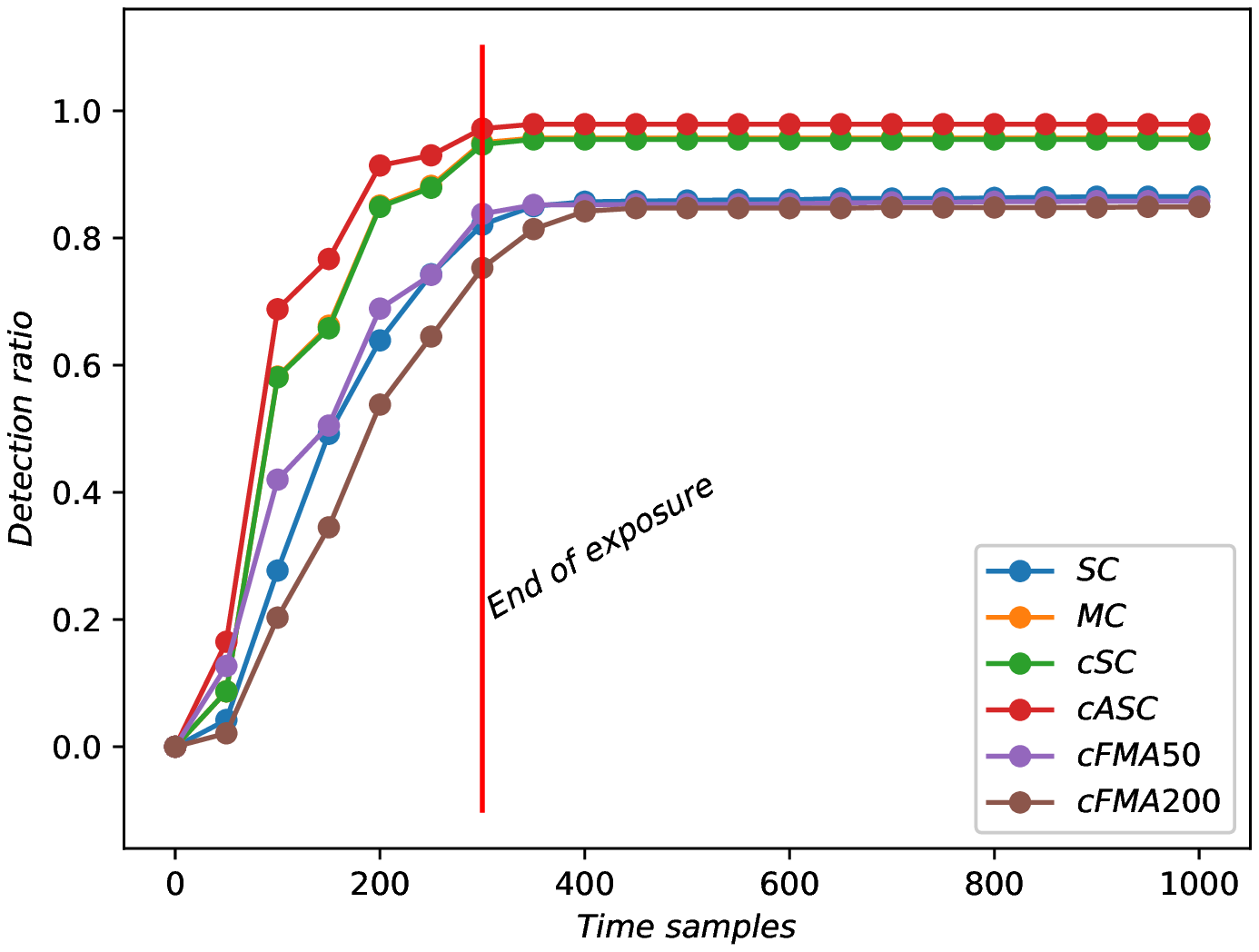} & \includegraphics[width=0.33\hsize]{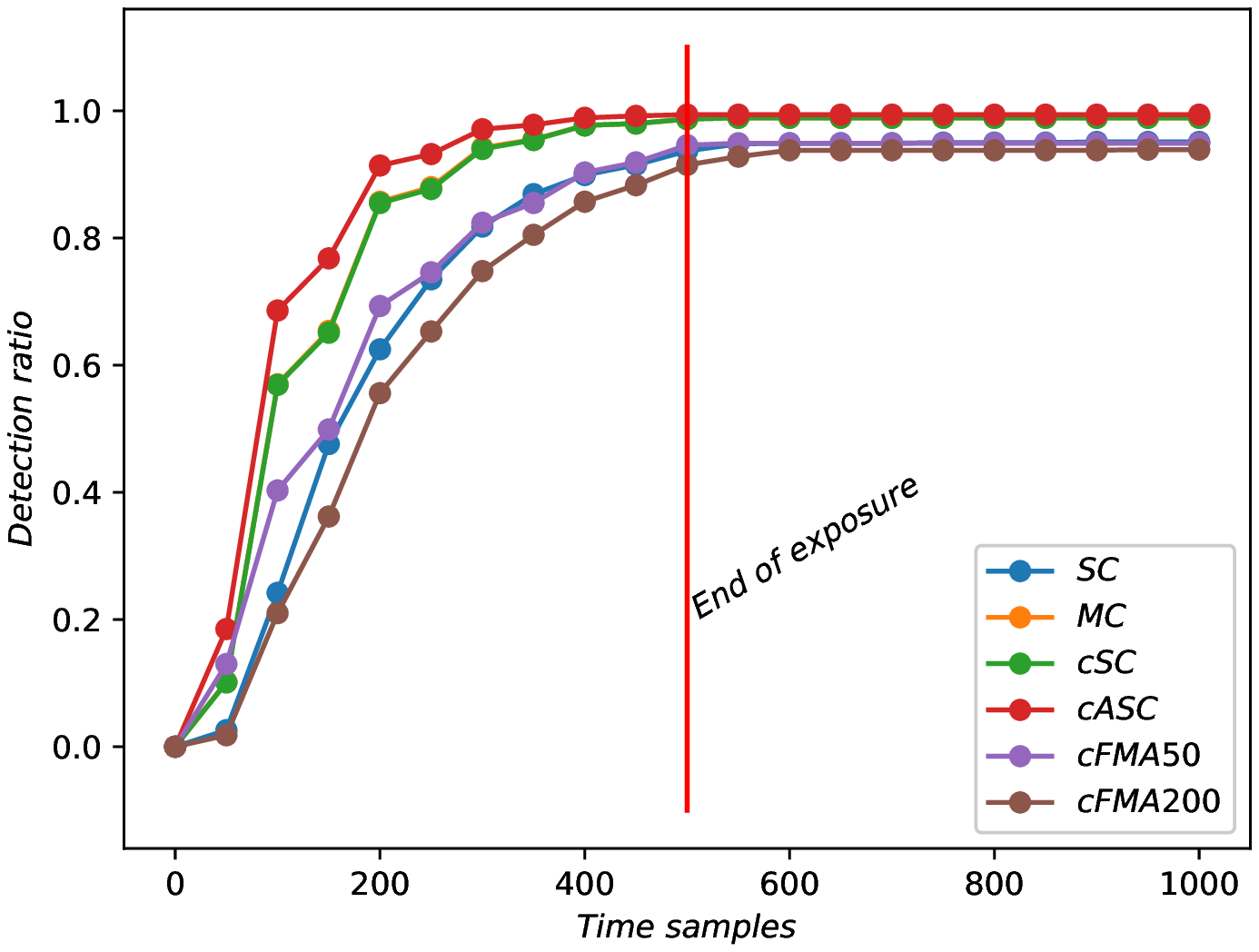}& 
    \includegraphics[width=0.33\hsize]{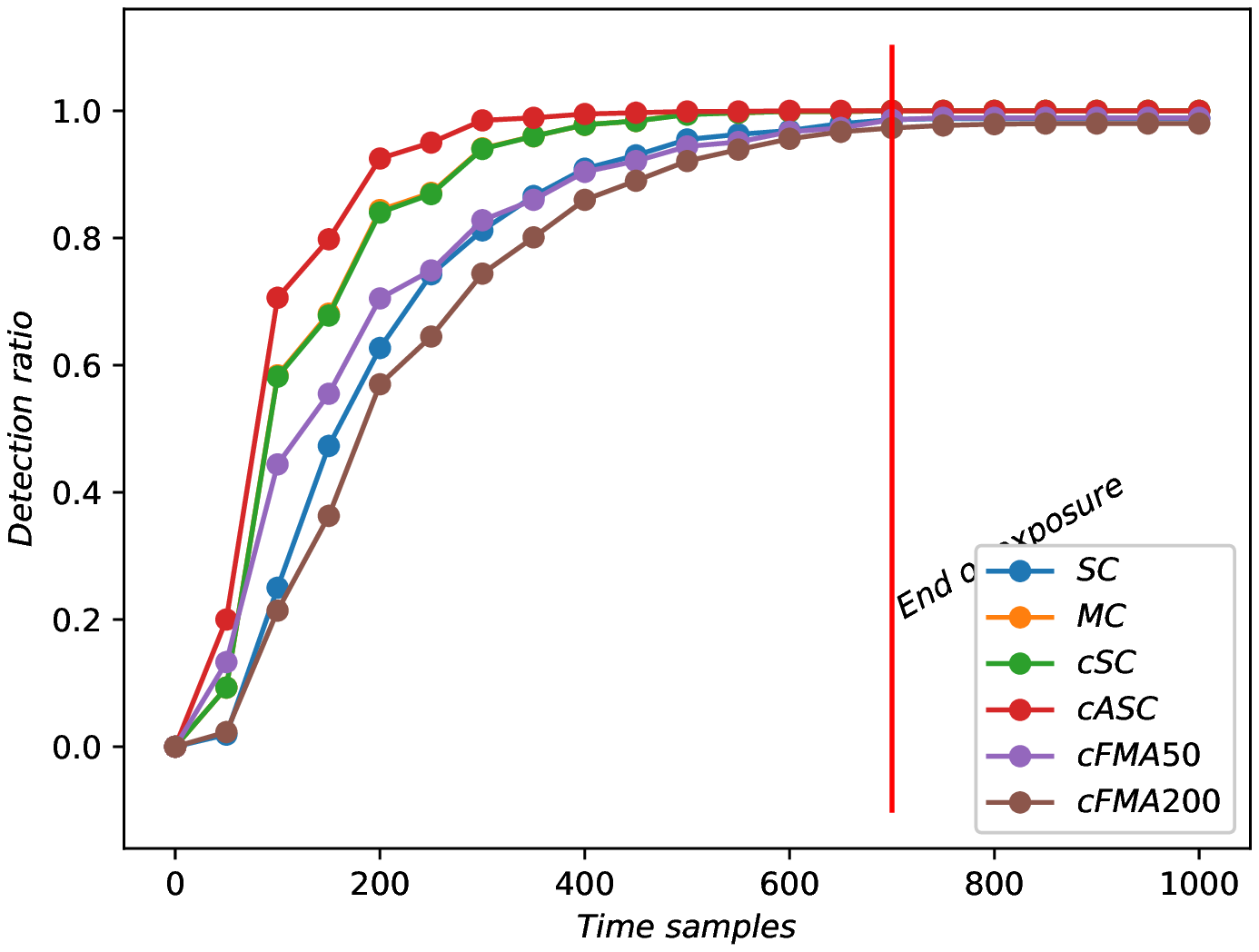} \\
    \includegraphics[width=0.33\hsize]{fig/OOOSync3_04.eps} & \includegraphics[width=0.33\hsize]{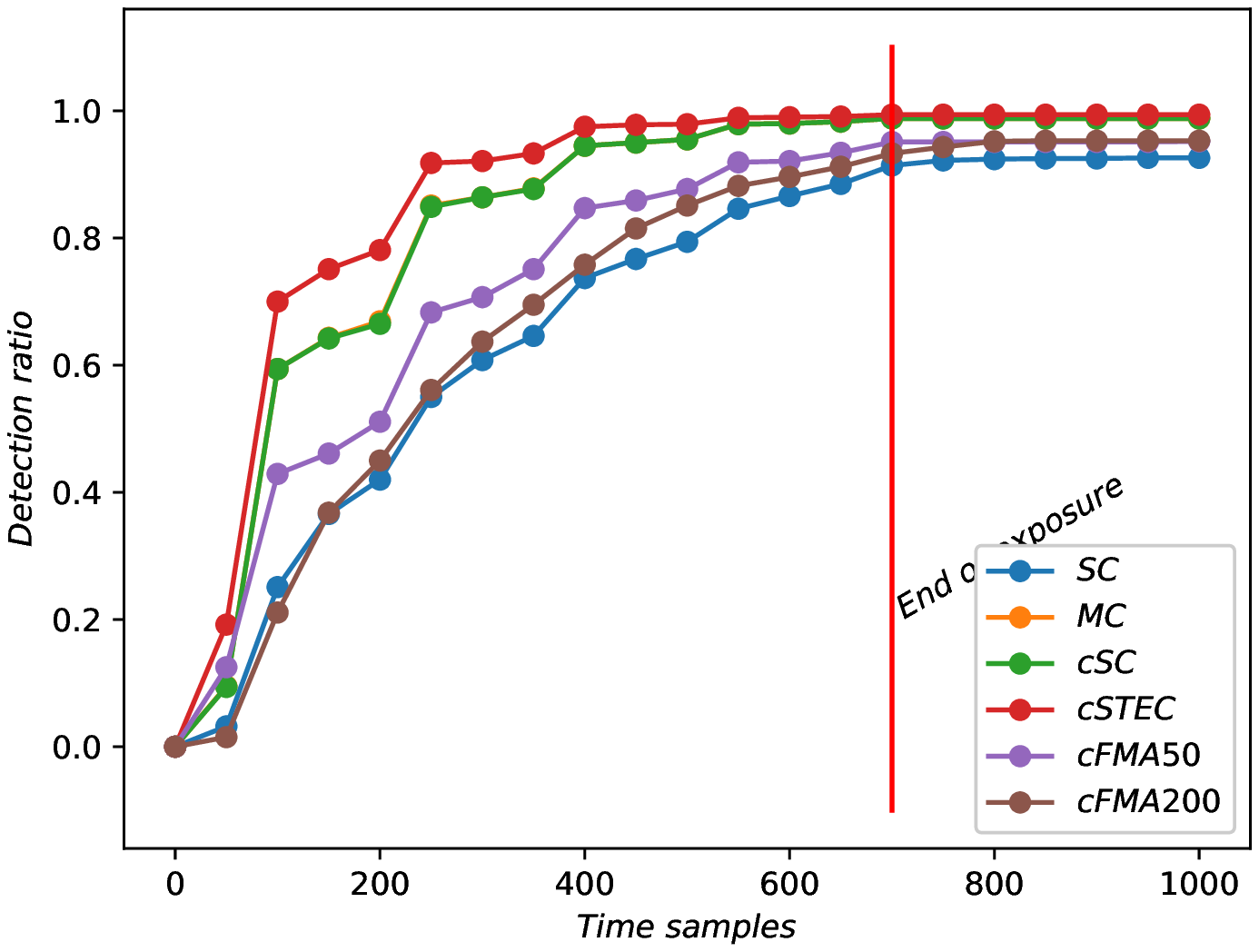}& 
    \includegraphics[width=0.33\hsize]{fig/OOOSync7_04.eps} \\
\end{tabular}
\caption{Results with 3/10 affected data streams on the first column, 5/10 on the second column and 7/10 on the third. With synchronised signals on the first row, a 50 time samples delay on the second, a full signal length delay on the third, and a 50 sample gap between two exposure on the fourth row. Signal amplitude is 0.4 and there is 100 time samples by exposure.}
\end{table}

\begin{table}[H]
\begin{tabular}{c c c}
    \includegraphics[width=0.33\hsize]{fig/Sync3_02.eps} & \includegraphics[width=0.33\hsize]{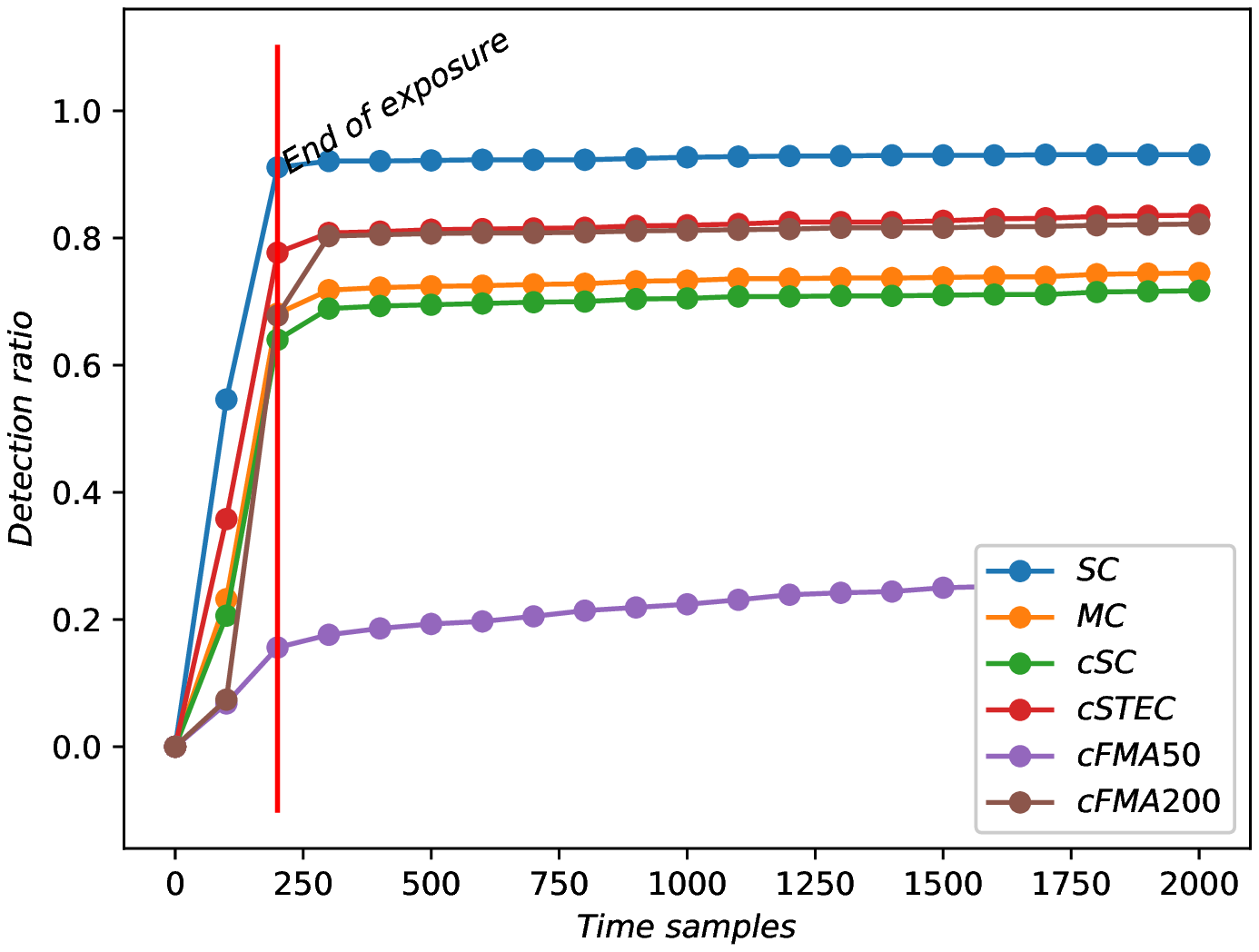}& 
    \includegraphics[width=0.33\hsize]{fig/Sync7_02.eps} \\
    \includegraphics[width=0.33\hsize]{fig/OSync3_02.eps} & \includegraphics[width=0.33\hsize]{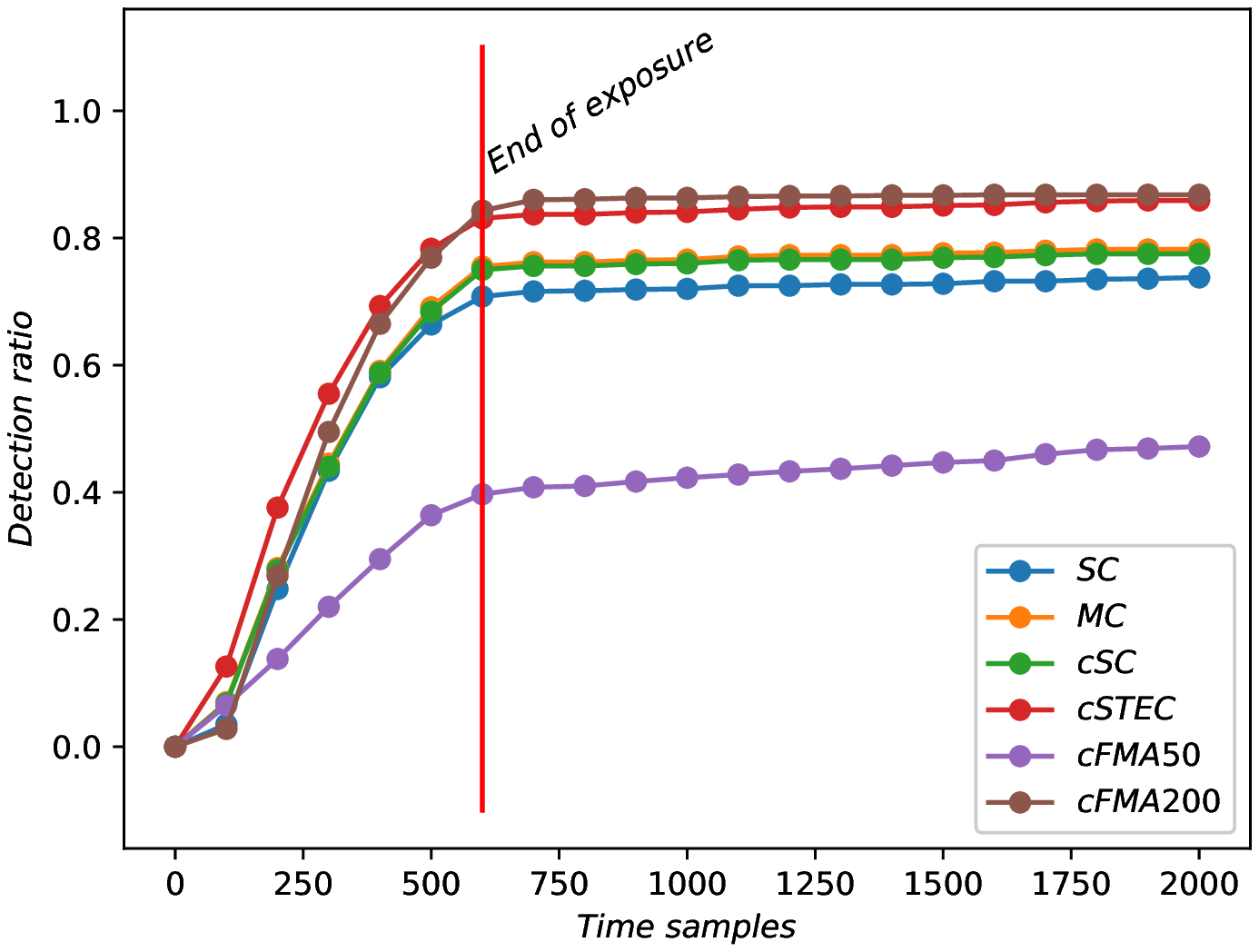}& 
    \includegraphics[width=0.33\hsize]{fig/OSync7_02.eps} \\
    \includegraphics[width=0.33\hsize]{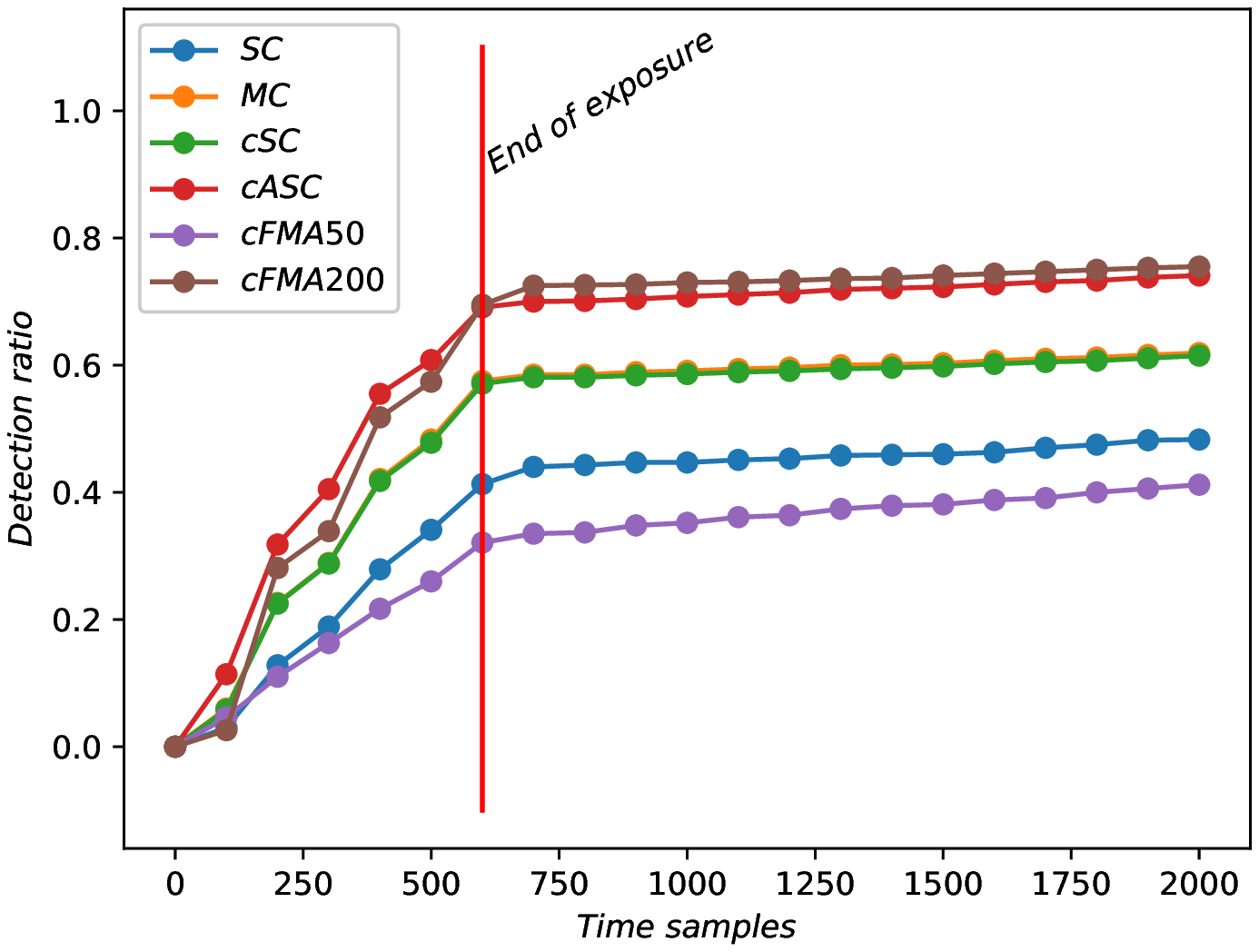} & \includegraphics[width=0.33\hsize]{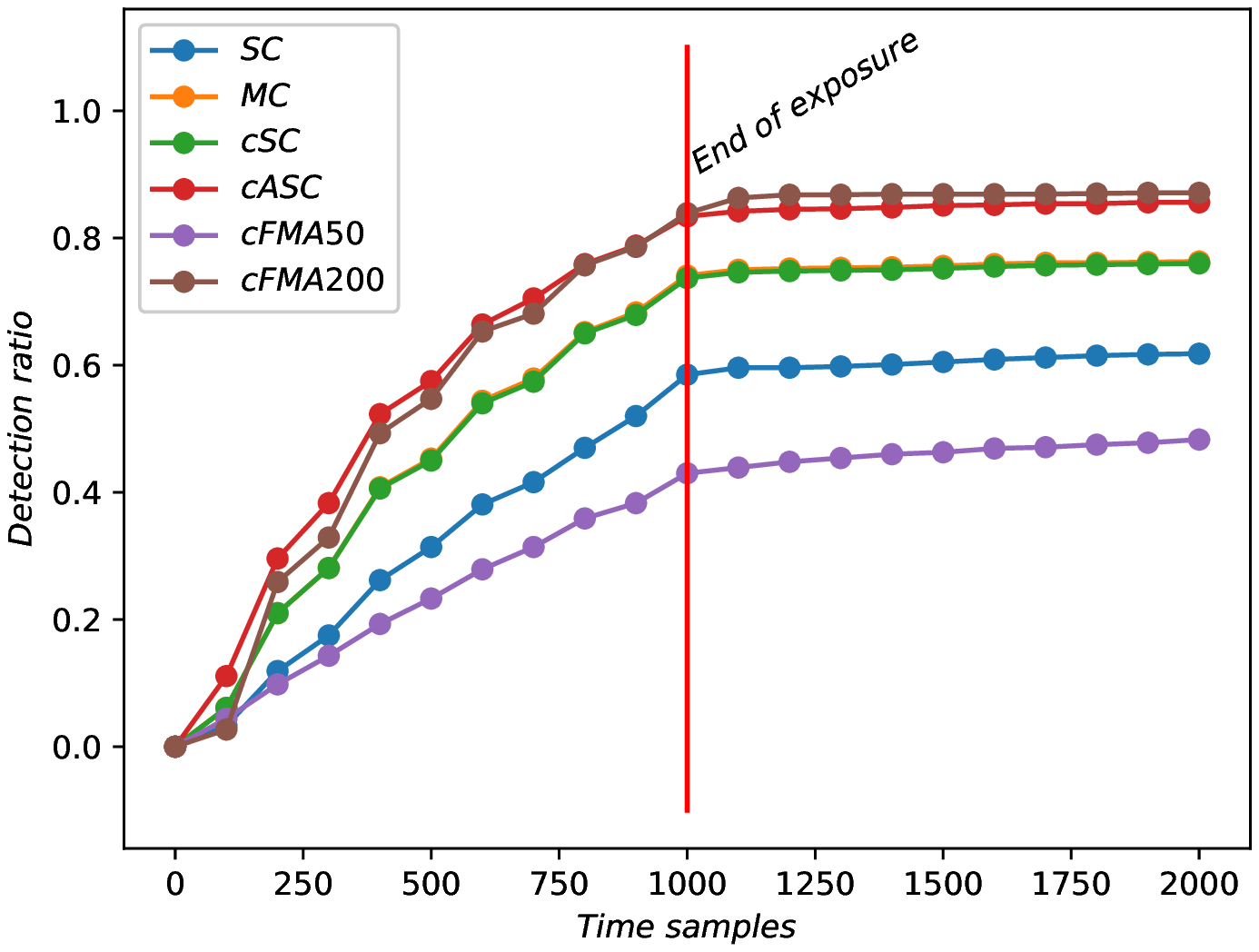}& 
    \includegraphics[width=0.33\hsize]{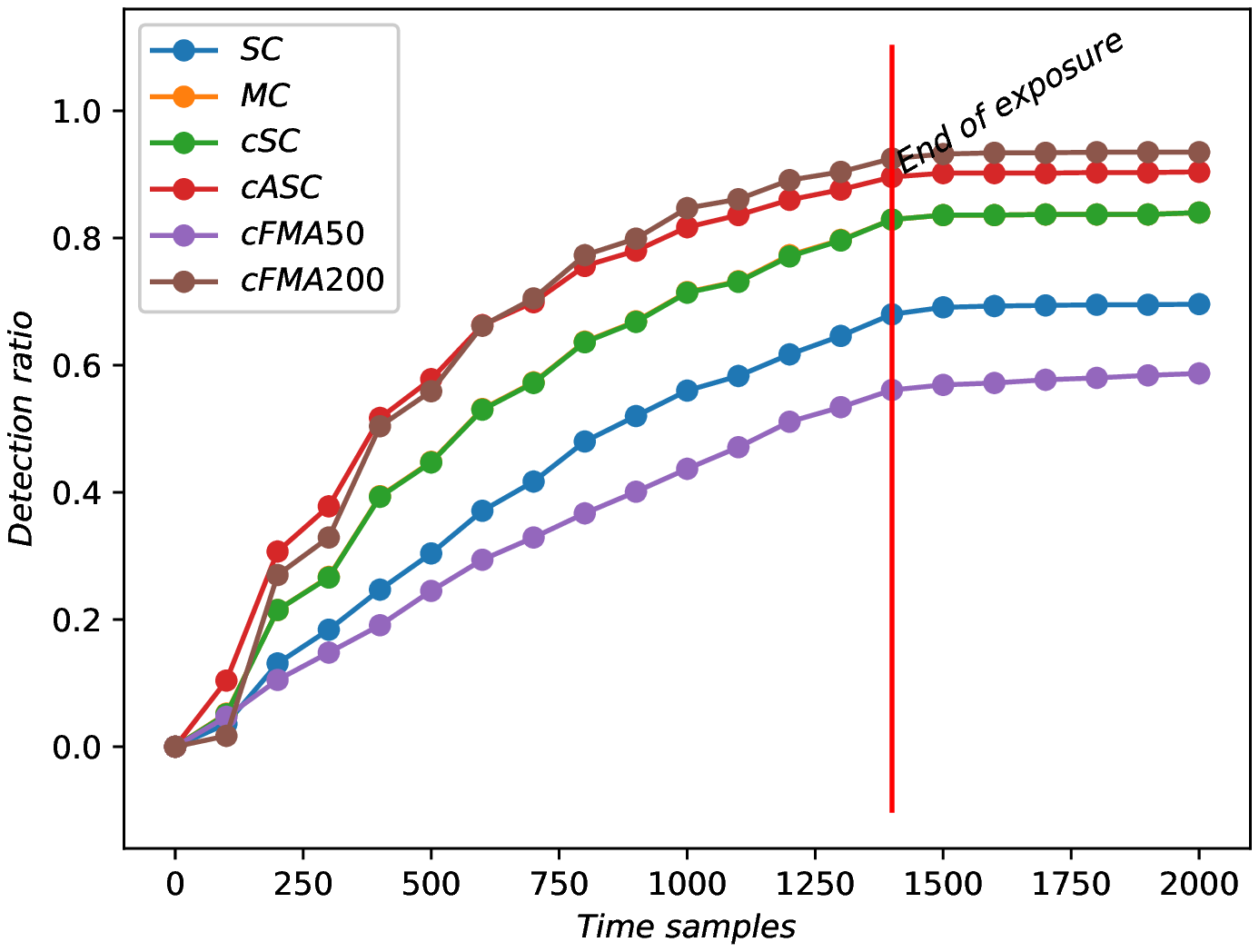} \\
    \includegraphics[width=0.33\hsize]{fig/OOOSync3_02.eps} & \includegraphics[width=0.33\hsize]{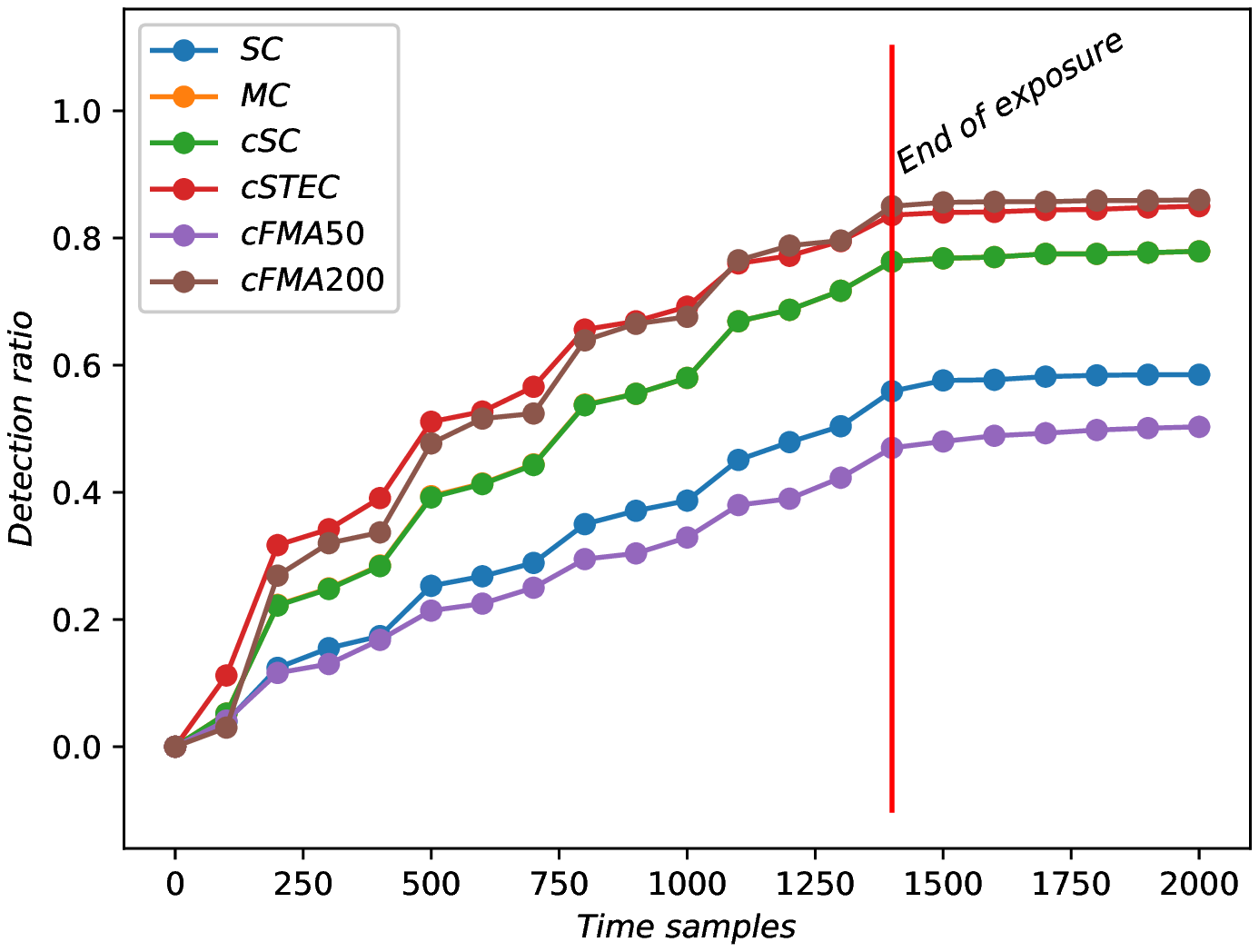}& 
    \includegraphics[width=0.33\hsize]{fig/OOOSync7_02.eps} \\
\end{tabular}
\caption{Results with 3/10 affected data streams on the first column, 5/10 on the second column and 7/10 on the third. With synchronised signals on the first row, a 100 time samples delay on the second, a full signal length delay on the third, and a 100 sample gap between two exposure on the fourth row. Signal amplitude is 0.2 and there is 200 time samples by exposure.}
\end{table}
\newpage
\begin{landscape}
\section{Computational resources}
\label{appendix2}

\begin{table}[H]
\hspace{0cm}
\begin{tabular}{|c| c| c| c|}
\hline
method & description & number of &  number of \\
 & & computations &  stored variables\\
\hline
MaxCUSUM & L access g + L tests  & 2L & L\\
\hline
SumCUSUM &  L access g + L sums & 2L & L\\
\hline
CensoredSC & L access g + L test + $<L$ sums   & $<3L$ & L\\
\hline
TE-CUSUM & same as cSC + L tests on G  & $<5L$ & 2L\\
\hline
FMA & $L\times 2$ (access + sums) + L tests + $<L$ sums&  $<6L$  & $w\times L$\\
\hline
\end{tabular}
\caption{Estimation of the computational cost of each method (L is the number of sensors and w is the window length for the FMA). In our cases $L << w$. Here the computation of the likelihood ratio is not displayed as it is common to all the techniques. The global decision test is not displayed for the same reason. \red{Computer operations have been categorised in 3 type of operation; memory access to local values of g, tests between two variables and sum. To simplify, every operation listed here are considered to have the same computational cost.}}
\end{table}
\vspace{-0.5cm}
As an example, we can describe the functioning of the TE-CUSUM based on the algorithm pseudo-code of \ref{appendix3}:

Dots $1,2$ and $7$ are common to all the techniques so these are not displayed in the comparison table.

Dots $1,2,3,4$ and $7^th$ are common to all the CUSUM based techniques. 

Dot $6$ is common to Censored Sum-CUSUM and Censored TE-CUSUM.

At every measure time:
\begin{itemize}
\vspace{-0.3cm}
    \item $L$ values are recorded by the sensors. 
    \vspace{-0.3cm}
    \item for each of these L values, the log-likelihood is computed
    \vspace{-0.3cm}
    \item $L$ values of W are loaded from the memory
    \vspace{-0.3cm}
    \item $L$ tests between W and W + L are performed to update W
    \vspace{-0.3cm}
    \item $L$ tests between G and W are performed to update G
    \vspace{-0.3cm}
    \item  $l<L$ Sum are performed to get the global test statistic
   \vspace{-0.3cm}
    \item the global test statistic is compared to a threshold to trigger or not detection
\end{itemize}

\end{landscape}

 \bibliographystyle{elsarticle-num} 
 \bibliography{main}

\end{document}